%% file: elgot-base-arxiv.tex
\documentclass[oribibl,envcountsame,final]{llncs}
\usepackage{graphicx}

\usepackage{etex}

\usepackage[pdfpagemode={UseOutlines},bookmarks=true,bookmarksopen=true,
   bookmarksopenlevel=0,bookmarksnumbered=true,hypertexnames=false,
   colorlinks,linkcolor={blue},citecolor={blue},urlcolor={red},
   pdfstartview={FitV},unicode,breaklinks=true,final]{hyperref}

\hypersetup{
    bookmarks=true,         
    unicode=false,          
    pdftoolbar=true,        
    pdfmenubar=true,        
    pdffitwindow=false,     
    pdftitle={Complete Elgot Monads and Coalgebraic Resumptions},    
    pdfauthor={Sergey Goncharov},     
    pdfsubject={},   
    pdfcreator={Sergey Goncharov},   
    pdfproducer={},  
    pdfkeywords={Elgot Monads,}, 
    pdfnewwindow=true,      
}

\usepackage{cite}

\usepackage{amssymb,amssymb,amsopn,amsmath}

\usepackage{array}
\usepackage{mathtools}
\usepackage{subcaption}
\usepackage{framed}
\usepackage{bm}	
\usepackage[bbgreekl]{mathbbol}	
\usepackage{float}
\usepackage{rotating}

\usepackage{enumitem}

\newenvironment{citemize}{\begin{itemize}[leftmargin=0cm,itemindent=.7cm,labelwidth=\itemindent,labelsep=0cm,align=left]}{\end{itemize}}

\DeclareFontFamily{U}{matha}{\hyphenchar\font45}
\DeclareFontShape{U}{matha}{m}{n}{
      <5> <6> <7> <8> <9> <10> gen * matha
      <10.95> matha10 <12> <14.4> <17.28> <20.74> <24.88> matha12
      }{}
\DeclareSymbolFont{matha}{U}{matha}{m}{n}
\DeclareMathSymbol{\HASH}{3}{matha}{"23}
\newcommand{\hash}{\operatorname{\HASH}}

\DeclareFontFamily{U}{MnSymbolC}{}
\DeclareFontShape{U}{MnSymbolC}{m}{n}{
    <-6>  MnSymbolC5
   <6-7>  MnSymbolC6
   <7-8>  MnSymbolC7
   <8-9>  MnSymbolC8
   <9-10> MnSymbolC9
  <10-12> MnSymbolC10
  <12->   MnSymbolC12}{}
\DeclareSymbolFont{MnSyC}{U}{MnSymbolC}{m}{n}
\newcommand{\filledsquare}{\operatorname{\scalebox{0.5}{$\blacksquare$}}}

\usepackage{tikz}
\usetikzlibrary{%
  fit,%
  calc,%
  arrows,%
  arrows.meta,%
  shapes.misc,
  shapes.arrows,%
  patterns,%
  automata,%
  chains,%
  matrix,%
  positioning,
  scopes,%
  decorations.pathmorphing
}
\usepackage{tikz-cd}

\usepackage{ifthen}
\newboolean{full}
\setboolean{full}{true}

\newcommand{\descto}[3][]{
    \arrow[draw=none]{#2}[description,fill=none,#1]{#3}
}

\usepackage{todos}

\newcommand{\klstar}{\star}  			
\newcommand{\klcomp}{\mathbin{\diamond}}  	
\newcommand{\istar}{\dagger}  			
\newcommand{\iistar}{\ddagger}  		
\newcommand{\comp}{\,}				

\newcommand{\IB}[2]{#1\mathbin{\hash}#2}

\newcommand{\B}[2]{#1\mathbin{\Box}#2}

\newcommand{\IBnu}{\digamma\kern-3pt{}_{\hash}}
\newcommand{\IBnun}{\digamma\kern-3pt{}_{\widehat\hash}}

\newcommand{\mun}{\mu^{\nu}}
\newcommand{\etan}{\eta^{\nu}}

\newcommand{\cbElg}{\catname{CElg}_{\hash}}

\newcommand{\SigF}{\Sigma}
\newcommand{\TF}{T_\SigF}
\newcommand{\BBTF}{\BBT_\SigF}

\newcommand{\ext}{\operatorname{\sf ext}}
\newcommand{\out}{\operatorname{\sf out}}
\newcommand{\tuo}[1][]{\operatorname{\sf out}_{#1}^{\text{\kern.5pt\rm-}\kern-.5pt1}\kern-1pt}

\newcommand{\coit}{\operatorname{\sf coit}}
\newcommand{\corec}{\operatorname{\sf coit}}

\newcommand{\kinl}{\ul{\operatorname{\sf inl}}}
\newcommand{\kinr}{\ul{\operatorname{\sf inr}}}

\makeatletter
\def\kto{{%
    \setbox0\hbox{$\longrightarrow$}%
    \rlap{\hbox to \wd0{$\hss\klcomp\kern1pt\hss$}}\box0
}}
\makeatother

\newcommand{\kpl}{\hm{[\kern-2.6pt[}}
\newcommand{\kpr}{\hm{]\kern-2.6pt]}}

\newcommand{\pmonad}{parameterized monad}
\newcommand{\Pmonad}{Parameterized monad}

\renewcommand{\paragraph}[1]{\textbf{#1}}

\makeatletter
\renewcommand{\xrightarrow}[2][]{\ext@arrow 0359\rightarrowfill@{~~#1~}{#2}}
\makeatother

\input{catprog}

\input{defslncs}

\usepackage{url}
\urldef{\mails}\path|{sergey.goncharov, christoph.rauch, stefan.milius}@fau.de|    

\usepackage[author=anonymous,nomargin,marginclue,footnote]{fixme}
\FXRegisterAuthor{sm}{asm}{SM}
\FXRegisterAuthor{sg}{asg}{SG}
\FXRegisterAuthor{cr}{acr}{CR}
\renewenvironment{definition}{\begin{defn}\rm}{\end{defn}}
\newenvironment{proof*}[1]{\begin{proof}[#1]\rm}{\end{proof}}

\usepackage{ifdraft}
\ifdraft
{
\usepackage[notcite,notref]{showkeys}
}{}

\newtheorem{constr}[thm]{Construction}
\def\ol{\overline}
\tikzset{shiftarr/.style={
        rounded corners,%
        to path={--([#1]\tikztostart.center)
                     -- ([#1]\tikztotarget.center) \tikztonodes
                     -- (\tikztotarget)},
}}

\begin{document}
\allowdisplaybreaks

\begin{frontmatter}
  \title{Complete Elgot Monads and Coalgebraic Resumptions}
  \author{Sergey Goncharov \thanks{Email: \href{mailto:Sergey.Goncharov@fau.de} {\texttt{ Sergey.Goncharov@fau.de }}}\and
  Stefan Milius \thanks{Email: \href{mailto:mail@stefan-milius.eu} {\texttt{ mail@stefan-milius.eu}} Supported by Deutsche Forschungsgemeinschaft (DFG) under project MI~717/5-1} \and
  Christoph Rauch \thanks{Email: \href{mailto:Christoph.Rauch@fau.de} {\texttt{ Christoph.Rauch@fau.de}}}}
\institute{Lehrstuhl f\"ur Theoretische Informatik, Friedrich-Alexander Universit\"at Erlangen-N\"urnberg, Germany}

\maketitle

\pagenumbering{arabic}

\begin{abstract}
\emph{Monads} are extensively used nowadays to abstractly model a wide range of computational effects such as nondeterminism, statefulness, and exceptions. It turns out that equipping a monad with a \emph{(uniform) iteration operator} satisfying a set of natural axioms allows for modelling \emph{iterative computations} just as abstractly. The emerging monads are called \emph{complete Elgot monads}. It has been shown recently that extending complete Elgot monads with free effects (e.g.\ operations of sending/receiving messages over channels) canonically leads to \emph{generalized coalgebraic resumption monads}, previously used as semantic domains for non-wellfounded guarded processes. In this paper, we continue the study of the relationship between abstract complete Elgot monads and those that capture coalgebraic resumptions, by comparing the corresponding categories of \emph{(Eilenberg-Moore) algebras}. To this end we first provide a characterization of the latter category; even more generally, we formulate this characterization in terms of Uustalu's \emph{parametrized monads}. This is further used for establishing a characterization of complete Elgot monads as precisely those monads whose algebras are coherently equipped with the structure of algebras of coalgebraic resumption monads.
\end{abstract}
\end{frontmatter}

\section{Introduction}
One traditional use of monads in computer science, stemming from the seminal thesis of 
Lawvere~\cite{Lawvere63}, is as a tool for algebraic
semantics where monads arise as a high-level metaphor for (clones of) equational theories.  
More recently, Moggi proposed to associate monads with \emph{computational effects}
and use them as a generic tool for denotational 
semantics~\cite{Moggi91b}, which later had a considerable impact on the design
of functional programming languages, most prominently Haskell~\cite{Haskell98}.
Finally, in the first decade of the new millennium, Plotkin and Power reestablished the connection
between computational monads and algebraic theories in their theory of \emph{algebraic 
effects}~\cite{PlotkinPower01,PlotkinPower02}.

We use the outlined perspective to study the notion of \emph{iteration}, a 
concept, that has a well-established algebraic description, 
and whose relevance in the context of computational effects is certain.
On the technically level our present work can be viewed as a continuation of the previous extensive work on monads with 
iteration~\cite{AczelAdamekEtAl03,AdamekMiliusEtAl06j,AdamekMiliusEtAl10} having its roots in
the work of Elgot~\cite{Elgot75} and Bloom and \'{E}sik~\cite{BloomEsik93} on iteration theories.

More specifically, we are concerned with a particular construction on monads:
given a monad $\BBT$ and a functor $\SigF$, we assume the existencence of the coalgebra
\begin{equation}\label{eq:TF}\tag{$\bigstar$}
  \TF X = \nu\gamma.\,T(X+\SigF\gamma)
\end{equation}
for each object $X$ (these final coalgebras exist under mild assumptions on~$T$,~$\Sigma$, and
the base category).  
It is known~\cite{Uustalu03} that~$\TF$ extends to a monad~$\BBTF$ and we call the latter
the \emph{generalized coalgebraic resumption monad}. 

Intuitively,~\eqref{eq:TF} is a generic 
semantic domain for systems combining \emph{extensional} (via~$\BBT$) and
\emph{intensional} (via~$\Sigma$) features with iteration. To make this intuition more precise, consider 
the following simplistic
\begin{expl}\label{expl:proc}
Let $A=\{a,b\}$ be an alphabet of \emph{actions}. Then the following system of
equations specifies \emph{processes} $x_1,x_2,x_3$ of \emph{basic process algebra (BPA)}:
\begin{align*}
x_1 = a\cdot (x_2 + x_3) && x_2 = a\cdot x_1 + b\cdot x_3 && x_3 = a\cdot x_1 + \checkmark
\end{align*}
We can think of this specification as a map $P\to T (\{\checkmark\}+\Sigma P)$
where $P=\{x_1,x_2,x_3\}$, $\Sigma=A\times\argument$ and $T=\PSet_{\omega}$ is the finite
powerset monad. Using the standard approach~\cite{RuttenTuri94} we can \emph{solve} this specification
by finding a map $P\to T_{\Sigma}\{\checkmark\}$ that assigns to every $x_i$ the
corresponding semantics over the domain of possibly non-wellfounded trees 
$T_{\Sigma}\{\checkmark\}=\nu\gamma.\PSet_{\omega}(\{\checkmark\}+A\times\gamma)$.
The crucial fact here is that the original system is \emph{guarded}, i.e.\ every
recursive call of a variable~$x_i$ is preceded by an action. In particular, this 
implies that the recursive system at hand has a unique solution. 
\end{expl}
If the guardedness assumption is dropped, the uniqueness of solutions can no longer
be ensured, but it is possible to introduce a notion of \emph{canonical solution}
assuming that the monad $\BBT$ has suitable completeness properties under an order,
or more generally is a \emph{complete Elgot monad}. A monad $\BBT$ is called a complete 
Elgot monad if it defines a \emph{solution} $f^\istar:X\to TY$ for every morphisms of the
form $f:X\to T(Y+X)$ satisfying a certain well-established set of axioms for iteration 
(e.g.\ $\PSet_{\omega}$ is not a complete Elgot monad, but the countable powerset monad
$\PSet_{\omega_1}$ is).
The central result of the recent work~\cite{GoncharovRauchEtAl15} is that whenever
$\BBT$ is a complete Elgot monad then so is the transformed monad~\eqref{eq:TF}. 
In particular, this allows for solving recursive equations over processes (in the 
sense of Example~\ref{expl:proc}) whenever recursive equations over 
$\BBT$ are solvable.

In the present paper we study the relationship between guarded and unguarded recursion
implemented via complete Elgot monads and generalized coalgebraic resumptions, respectively.
As an auxiliary abstraction device, we involve the notion of \emph{parametrized monad}
previously developed by Uustalu~\cite{Uustalu03}, e.g.\ the bifunctor $\IB{X}{Y}=T(X+\Sigma Y)$
in~\eqref{eq:TF} is a parametrized monad.

The paper is organized as follows. After categorical preliminaries (Section~\ref{sec:prelim}) 
we present and discuss complete Elgot monads in Section~\ref{sec:emon}. In 
Section~\ref{sec:cealg} we introduce algebras and complete Elgot algebras for
parametrized monads; here we show that existence of free complete Elgot algebras
is equivalent to the existence of certain final coalgebras, which then form carriers
of the corresponding algebras (Theorem~\ref{thm:eq}); furthermore, we show that 
the category of complete Elgot algebras is equivalent to the Eilenberg-Moore 
category of a generalized coalgebraic resumption monad over the corresponding parametrized monad 
(Theorem~\ref{thm:algebras-iso}). Finally, in Section~\ref{sec:alg} we apply the
developed results to characterize complete Elgot monads as those whose algebras
are coherently equipped with complete Elgot algebra structures (Theorem~\ref{thm:emon} and~\ref{thm:elgot_from_alg}).

\iffull\else
Due to space constraints we omit all proofs. A full version of our paper can be found on arXiv.
\fi

\section{Preliminaries}\label{sec:prelim}

We assume that readers are familiar with basic category theory~\cite{Maclane71}; we
write $|\BC|$ for the class of objects of a category $\BC$ and
$f:X\to Y$ for morphisms in $\BC$. We often omit indexes, e.g.\ on natural transformations, if they are
clear from the context. 

In this paper we work with an ambient category $\BC$ with finite coproducts.
We denote by $\inl$ and $\inr$ the left- and right-hand coproduct injections from $X$ and $Y$ to
$X+Y$, 
and $[f,g]:X+Y\to Z$ the is the \emph{copair} of $f:X\to Z$ and
$g:Y\to Z$, i.e.~the unique morphism with $[f,g]\inl = f$ and
$[f,g]\inr = g$. The codiagonal is denoted by $\nabla = [\id,\id]:X+X\to X$ as usual.

We consider \emph{monads} by $\BC$ given in the form of \emph{Kleisli triples}
$\BBT=(T,\eta,\argument^{\klstar})$ where $T$ is an endomap on $|\BC|$, 
$\eta$, called \emph{monad unit}, is a family of morphisms ${\eta_X:X\to TX}$ 
indexed over $|\BC|$, and \emph{(Kleisli) lifting} assigning to each $f: X \to TY$ a morphism $f^\klstar: TX \to TY$ such that the following laws hold:
\begin{align*} 
\eta^{\klstar}=\id, && f^{\klstar}\comp\eta=f, && (f^{\klstar}\comp g)^{\klstar}=f^{\klstar}\comp g^{\klstar}.
\end{align*}
This is equivalent to the definition of a monad in terms of \emph{monad
  multiplication}~$\mu$~\cite{Maclane71}, where in particular $\mu=\id^\klstar$,
$\eta$ extends to a natural transformation, and $T$ to an endofunctor by
$Tf=(\eta\comp f)^\klstar$. The \emph{Kleisli category} $\BC_{\BBT}$ of $\BBT$
is formed by \emph{Kleisli morphisms}: $\Hom_{\BC_{\BBT}}(X,Y)=\Hom_{\BC}(X,TY)$
under $\eta$ used as identity morphisms and \emph{Kleisli composition}
$(f,g)\mapsto f\klcomp g= f^\klstar\comp g$. We adopt the notation $f:X\kto Y$
for Kleisli morphisms $f:X\to TY$.

The forgetful
functor from $\BC_{\BBT}$ to $\BC$ has a left adjoint sending any \mbox{$f:X\to Y$} to $\ul{f}=\eta\comp f:X\to TY$. 
Like any left adjoint, this functor preserves colimits, and 
in particular coproducts. Since $|\BC|=|\BC_{\BBT}|$, this implies that coproducts 
in $\BC_{\BBT}$ exist and are lifted from $\BC$. Explicitly, $\kinl=\eta\inl:X\kto X+Y$,
$\kinr=\eta\inr:X\kto X+Y$ are the coproduct injections in $\BC_{\BBT}$ and 
$[f,g]:A+B\kto C$ is the copair of $A\kto C$ and $B\kto C$. We denote by 
$f\oplus g:A+B\kto A'+B'$ the coproduct of morphisms $f:A\kto A'$ and $g:B\kto B'$ in
$\BC_{\BBT}$. Besides $\BC_{\BBT}$, we consider the category $\BC^{\BBT}$ 
of \emph{(Eilenberg-Moore) algebras} for $\BBT$, whose objects are pairs $(A,a:TA\to A)$, satisfying laws: $a\comp\eta=\id$, $a\comp (Ta)=a\comp\mu$;
see~\cite{Maclane71} for more details.

We call on the standard facts on \emph{($F$-)coalgebras}~\cite{Rutten96}, which are pairs 
of the form $(X,f)$ with \emph{carriers} $X$ ranging over $|\BC|$ and 
\emph{transition structures} $f$ ranging over $\Hom_{\BC}(X,FX)$ for a fixed endofunctor $F$. 
Coalgebras together with morphisms of the carriers compatible
with the transition structure form a category. The final $F$-coalgebra, if it exists, is denoted
$(\nu F,\out)$. By Lambek's lemma, $\out$ is an isomorphism, whose inverse $\tuo$ 
can be obtainend as $\coit(F\out)$ where for any coalgebra $(X,f:X\to FX)$ we denote by $\coit f$ the 
unique coalgebra morphism $X\to\nu F$ to the final coalgebra.
 
\input{diagram}
\section{Complete Elgot Monads for Iteration}\label{sec:emon}
Complete Elgot monads are a slight generalization of Elgot monads from~\cite{AdamekMiliusEtAl10,JirAdamekMiliusEtAl11}, which in turn, for the base category being $\Set$, correspond precisely to those iteration theories of Bloom and \'Esik~\cite{BloomEsik93} that satisfy the functorial dagger implication for base morphisms. 
In the following definition cited 
from~\cite{GoncharovRauchEtAl15}, we 
follow the terminology of~\cite{SimpsonPlotkin00,BentonHyland03} 
where the same axioms were considered in the dual setting of generic parametrized recursion.
\begin{defn}[Complete Elgot monads]\label{defn:elgot}
  A \emph{complete Elgot monad} is a monad $\BBT$ equipped with an
  operator $\argument^\istar$, called \emph{iteration}, that assigns
  to each morphism $f:X\kto Y+X$ a morphism $f^{\istar}:X\kto Y$
  such that the following axioms hold:
\begin{itemize}
  \item\emph{fixpoint:} $f^{\istar}=[\eta, f^{\istar}]\klcomp f$, for $f:X\kto Y+X$;
  \item\emph{naturality:} $g \klcomp f^{\istar} = ((g\oplus\eta)\klcomp f)^{\istar}$ for $g : Y \kto Z$;
  \item\emph{codiagonal\footnote{The codiagonal axiom is often written as $((\eta\oplus\ul\nabla) \klcomp g)^{\istar} = (g^{\istar})^{\istar}$ implicitly alluding to the canonical isomorphism $Y+(X+X)\cong (Y+X)+X$.}:} $([\eta,\kinr] \klcomp g)^{\istar} = (g^{\istar})^{\istar}$ for  $g : X \kto (Y + X) + X$;
  \item\emph{uniformity:} $f \klcomp \ul{h} = (\eta\oplus\ul{h}) \klcomp g$ implies
	$f^{\istar} \klcomp \ul{h} = g^{\istar}$ for  $g: Z \kto Y + Z$ and
	$h: Z \to X$.
\end{itemize}
\end{defn}
The above axioms of iteration can be comprehensibly represented in a flowchart-style
as in Fig.~\ref{fig:ax}. Here the feedback loops correspond to iteration and the
colored frames indicate the scope of the constructs being iterated. 
We believe that this presentation illustrates that these axioms are natural and desirable laws of iteration. For example, the naturality axiom expresses the
fact that the scope of the iteration can be stretched to embrace a function 
postprocessing  the output of the terminating branch. There is an obvious similarity
between the axioms in Fig.~\ref{fig:ax} and the axioms of \emph{traced monoidal 
categories}~\cite{JoyalStreetEtAl96}. In fact, Hasegawa~\cite{Hasegawa97} proved that there is an equivalent presentation of a dagger operation satisfying the above axioms in terms of a uniform trace operator w.r.t.~coproducts (actually, Hasegawa worked in the dual setting with products). 
Note that the present axioms make use of coproduct injections and the codiagonal morphism, while the trace axioms can be formulated more generally for any monoidal product.

One standard source of examples for complete Elgot monads is a suitable enrichment
of the Kleisli category $\BC_{\BBT}$ over complete partial orders. 
\begin{expl}\textbf{($\omega$-continuous monads)}\label{def:omega-cont}
  An \emph{$\omega$-continuous} monad consists of a monad $\BBT$ such
  that the Kleisli category $\BC_{\BBT}$ is enriched over the category
  $\Cppo$ of $\omega$-complete partial orders
  with bottom $\bot$ and
  (nonstrict) continuous maps; moreover, composition in $\BC$ is
  required to be left strict and composition in $\BC_{\BBT}$ right
  strict: $\bot\comp f=\bot$, $f \klcomp\bot=\bot$; equivalently,
  $\bot$ is a \emph{constant} of $\BBT$. Note that it follows that
  copairing in $\BC_{\BBT}$ is continous in both arguments; for
  $\bigsqcup_i [f_i, g]$ is a morphism satisfying
  $(\bigsqcup_i [f_i, g])\comp \kinl = \bigsqcup_i f_i$ and
  $(\bigsqcup_i [f_i, g])\comp \kinr = g$ by the continuity of composition,
  whence $\bigsqcup_i [f_i, g] = [\bigsqcup_i f_i , g]$ (and similarly
  for continuity in the second argument).

It is shown in~\cite{GoncharovRauchEtAl15} that an $\omega$-continuous monad is
a complete Elgot monad with $e^\istar$ calculated as the least fixed point of the
map $f\mapsto [\eta,f]\klcomp e$. This yields the powerset monad $\PSet$, the \emph{Maybe-monad}
$(\argument+1)$, or the nondeterministic state monad $\PSet(\argument\times S)^S$
as examples of complete Elgot monads on $\Set$. The \emph{lifting monad} $(\argument)_{\bot}$
is a complete Elgot monad on the category of complete partial orders without bottom.
\end{expl}
Another principal source of examples are free complete Elgot monads for which
the iteration of guarded morphisms is uniquely defined.
\begin{expl}\textbf{(Free complete Elgot monads)}
Suppose $\BBT$ is the initial complete Elgot monad. It is shown in~\cite{GoncharovRauchEtAl15} that
whenever the functor $\BBT_{\Sigma}$ defined by~\eqref{eq:TF} exists, it yields the 
\emph{free complete Elgot monad on $\Sigma$} (note that the original $\BBT$ is the
free complete Elgot monads on $\Sigma$ being the constant functor on the initial object of~$\BC$). 
On $\Set$ (more generally, on any \emph{hyperextensive category}~\cite{AdamekBorgerEtAl08}) 
the initial complete Elgot monad $\BBT$ is the Maybe-monad $\argument+1$.
\end{expl}
In comparison to the previous work~\cite{GoncharovRauchEtAl15},
Definition~\ref{defn:elgot} remarkably drops the axiom of
\emph{dinaturality} (see Fig.~\ref{fig:dina}). The reason for it is
that this axiom turns out to be derivable, which is a fact that was
recently discovered and formalized on the level of abstract iteration
theories~\cite{EsikGoncharov16}. Corollary~6 from \emph{op.cit.}~can
be couched in present terms (modulo the terminological change:
\emph{parameter identity} instead of \emph{naturality}, \emph{double
  dagger} instead of \emph{codiagonal} and \emph{dagger implication
  for base morphisms} instead of \emph{uniformity}) as follows:
\input{dinaturality}
%
%
%
%
%
%
%
%
%
%
%
%
%
%
%
%
%
\begin{prop}[Dinaturality]\label{prop:dinaturality-from-double}
Given $g : X \kto Y + Z$ and $h:Z\kto Y+X$, then
\[([\kinl, h] \klcomp g)^{\istar} = [\eta, ([\kinl,g] \klcomp h)^{\istar}]\klcomp g\]  
\end{prop}
The codiagonal axiom in Definition~\ref{defn:elgot} can equivalently be replaced 
by a form of the well-known \emph{Beki\'{c} identity}, see~\cite{BloomEsik93}.
\begin{prop}[Beki\'{c} identity]\label{prop:bekic-from-elgot}
  A complete Elgot monad $\BBT$ is, equivalently, a monad satisfying the
  \emph{fixpoint}, \emph{naturality} and \emph{uniformity} axioms (as in
  Definition~\ref{defn:elgot}), and the \emph{Beki\'{c} identity}
  \begin{align*}
   (T\alpha \comp [f, g])^{\istar} = [\eta,
    h^{\istar}] \klcomp [\kinr, g^{\istar}]
  \end{align*}
  where $g : X \kto (Z + Y) + X$, $f : Y \kto (Z + Y) + X$,
  $h = [\eta, g^{\istar}] \klcomp f : Y \kto Z + Y$, with
  $\alpha : (A + B) + C \to A + (B + C)$ being the obvious
  associativity morphism.
\end{prop}

\section{Parametrized Monads and Complete Elgot Algebras}\label{sec:cealg}
In order to study complete Elgot monads and their algebras it is helpful to make a further
abstraction step and generalize from monads to \emph{parametrized monads}~\cite{Uustalu03} (finitary parametrized monads are also called \emph{bases}~\cite{AdamekMiliusEtAl05}), which are of independent interest.
\begin{definition}\textbf{(\Pmonad)}
  A \emph{\pmonad} over $\BC$ is a functor from $\BC$ to the category of monads
  over $\BC$ and monad morphisms. More explicitly, a \pmonad{} is a bifunctor
  $\IB{}{}:\BC\times\BC\to\BC$ such that for any $X\in|\BC|$,
  $\IB{\argument}{X}:\BC\to\BC$ is a monad, and for any $f:X\to Y$, the family
  $(\IB{\id_Z}{f})_{Z}$ yields a monad morphism from $\IB{\argument}{X}$ to
  $\IB{\argument}{Y}$.
\end{definition} 
\begin{remark}
  The order of arguments in $\IB{X}{Y}$ is in agreement with~\cite{Uustalu03}
  and differs from~\cite{AdamekMiliusEtAl05} where the notation $\B{Y}{X}$
  equivalent to the present $\IB{X}{Y}$ is used. We chose the order of arguments
  to ensure agreement with the type profile of the iteration operator
  $\argument^\istar$, which is in turn in agreement with the
  expression~\eqref{eq:TF}.
\end{remark}
Following~\cite{AdamekMiliusEtAl05} we will from now on denote the unit and monad multiplication of monads
$\IB{\argument}{X}$ by $u_A^X : A \to \IB{A}{X}$ and $m_A^X :
\IB{(\IB{A}{X})}{X} \to \IB{A}{X}$, respectively. 
\begin{expl}\textbf{(Parametrized monads)}\label{ex:pmons}
We recall some standard examples of parametrized monads from~\cite{Uustalu03}; further examples can be found e.g.~in~\cite{JirAdamekMiliusEtAl08}.\smnote{If we have space left, I might include those explicitly.}
\begin{enumerate}
 \item Whenever $\BBT=(T,\eta,\argument^\klstar)$ is a monad and $\Sigma$ is a functor, $\IB{A}{X}=T(A+\Sigma X)$
   is a parametrized monad with the unit given by 
   \[
     u^X_A = \Bigl(A \xrightarrow{\inl} A + \Sigma X \xrightarrow{\eta_{A+\Sigma X}} T(A + \Sigma X)\Bigr)
   \]
   and the multiplication by
   \[
     m^X_A = \Bigl(T(T(A + \Sigma X) + \Sigma X) \xrightarrow{[\id,\eta_{A + \Sigma X}\inr]^\klstar} T(A + \Sigma X)\Bigr).
   \]
   Specifically, if $\Sigma$ is the constant functor on object $E$ then
   $\IB{X}{Y}$ is the exception monad transformer with exceptions from
   $E$~\cite{Moggi91b}. Another interesting special case is when $\BBT$ is the identity monad (cf.~Remark~\ref{rem:classicElgotalg}).
 \item $\IB{A}{X}=A\times X^\star$ is a parametrized monad with the unit and multiplication given by
   \[
     u^X_A: a \mapsto (a, \varepsilon) \quad\text{and}\quad m^X_A: (a, w, v) \mapsto (a, wv),
   \]
   where $\varepsilon$ denotes the empty word and $wv$ concatenation of words. 
 \item Given a contravariant endfunctor $H$, $\IB{A}{X}=A^{HX}$ is a parametrized monad with the unit and multiplication given by
   \[
     u^X_A: a \mapsto \lambda x.\,a \quad\text{and}\quad m^X_A: (f:HX\to (HX\to A)) \mapsto \lambda x.\, f(x)(x).
   \]
This is a generalization of the well known \emph{reader monad}, which can be recovered by instantiating $H$ with a constant functor.
\end{enumerate}
\end{expl}
The following is a straightforward extension of the notion of an algebra for a base studied in~\cite{AdamekMiliusEtAl05} to arbitrary parametrized monads.
\begin{definition}\textbf{($\IB{}{}$-algebras)}
  Given a \pmonad{} $\IB{}{}:\BC\times\BC\to\BC$, a \emph{$\IB{}{}$-algebra} is a pair $(A, a )$ consisting of an
  object $A$ of $\BC$, and an algebra for the monad $\IB{\argument}{A}$, i.e.\ a
  morphism $a : \IB{A}{A} \to A$ satisfying
  \[
    \begin{tikzcd}
      A
        \ar[r,  "u_A^A"]
        \ar[dr, "\id", swap,equal] &
      \IB{A}{A}
        \ar[d, "a "]\\
      & A
    \end{tikzcd}
    \hspace{1cm}
    \begin{tikzcd}[column sep=huge]
      \IB{(\IB{A}{A})}{A}
        \ar[r, "\IB{a}{\id}"]
        \ar[d, "m_A^A",swap] &
      \IB{A}{A}
        \ar[d, "a "] \\
      \IB{A}{A}
        \ar[r, "a "] &
      A
    \end{tikzcd}
  \]
A morphism between $\hash$-algebras $(A, a)$ and $(B, b)$ is a $\BC$-morphism
$f : A \to B$ such that $f \comp a = b \comp (\IB{f}{f})$.
\end{definition}
\iffull
\begin{expl}\label{ex:simple}
  Several examples of $\IB{}{}$-algebras have been discussed
  in~\cite{JirAdamekMiliusEtAl07,JirAdamekMiliusEtAl08}. Here we
  recall from \emph{loc.~cit.}~only the following. Consider the three
  bases $A \mathbin{\#_1} X = A + X \times X$,
  $A \mathbin{\#_2} X = A \times X^*$, and $A \mathbin{\#_3} X = BA$
  on $\Set$ where $BA$ is the free algebra with one binary operation
  on $A$ (i.e.~$BA$ consists of all finite binary trees with leaves
  labelled in $A$). Note that $\#_1$ and $\#_3$ are special cases of
  the \pmonad\ of Example~\ref{ex:pmons}(i) for $\BBT$ the indentity monad and
  $\Sigma X = X \times X$ and $\Sigma X = \emptyset$, respectively.
  The category of algebras is in each of the three cases isomorphic to
  the category of algebras with one binary operation. Later, when we
  discuss complete Elgot $\IB{}{}$-algebras, we are going to see a
  difference between these three \pmonad.
\end{expl}
\fi
For our leading example $\IB{X}{Y}=T(X+\Sigma Y)$ the category of $\IB{}{}$-algebras
can be described explicitly.
\begin{prop}\label{prop:T-Sigma-bialg}
Let $\IB{X}{Y}=T(X+\Sigma Y)$ for a monad\/ $\BBT$ and a functor $\Sigma$ on $\BC$. Then $\IB{}{}$-algebras
are precisely $\BBT$-$\Sigma$-bialgebras, i.e.~triples $(A,a,f)$ 
where $a:TA\to A$ is a $\BBT$-algebra and $f:\Sigma A\to A$ is a $\Sigma$-algebra.
\end{prop}
\begin{corollary}\label{cor:T-Sigma-bialg}
Let $\IB{X}{Y}=T(X+Y)$ for a monad\/ $\BBT$ on $\BC$. The category $\BC^{\BBT}$ 
of\/ $\BBT$-algebras is isomorphic to the full subcategory of those $\IB{}{}$-algebras 
$a:T(A+A)\to A$, which factor through $T\nabla$. 
\end{corollary}
Analogously to the case of monads, we introduce $\IB{}{}$-algebras with iteration. This generalizes the definition of a \emph{complete Elgot algebra for a functor} from~\cite{AdamekMiliusEtAl06j}.
\begin{definition}\textbf{(Complete Elgot $\IB{}{}$-algebras)}
  A \textit{complete Elgot $\IB{}{}$-algebra} is a $\IB{}{}$-algebra $a  : \IB{A}{A}
  \to A$ equipped with an iteration operator
  \[
    \frac{e : X \to \IB{A}{X}}{e^{\istar} : X \to A}
  \]
  satisfying the following axioms:
  \begin{itemize}
  \item\emph{solution:} for every $e: X \to \IB{A}{X}$ we have $e^{\istar} = a  \comp (\IB{\id}{e^{\istar}}) \comp e$;
  \item\emph{functoriality:} for every $e: X \to \IB{A}{X}$, $f: Y \to \IB{A}{X}$ and $h: X \to Y$, 
$f\comp h = (\IB{\id}{h}) \comp e$ implies $f^{\istar} \comp h = e^{\istar}$;
  \item\emph{compositionality:} for every $f: Y \to \IB{A}{Y}$ and $g: X \to \IB{Y}{X}$ define
    \[
      f^{\istar} \bullet g = (X \xrightarrow{g} \IB{Y}{X} \xrightarrow{\IB{f^\istar}{\id}} \IB{A}{X})
    \]
    and $f \filledsquare g: Y+X \to \IB{A}{(Y+X)}$ by
    \[
      \begin{tikzcd}
        Y+X \ar[r, "{[u^X_Y,g]}"]
        &
        \IB{Y}{X}
        \ar[r, "\IB{f}{\id}"]
        &
        \IB{(\IB{A}{Y})}{X}
        \ar[d, "\IB{(\IB{\id}{\inl})}{\inr}"]
        \\
        \IB{A}{(Y+X)}
        &&
        \IB{(\IB{A}{(Y+X)})}{(Y+X)}
        \ar[ll, "m^{Y+X}_A"']
      \end{tikzcd}
    \]
    Compositionality states that $(f \filledsquare g)^{\istar} \comp \inr = (f^{\istar} \bullet g)^{\istar}: X \to A$.
  \end{itemize}
    A \emph{morphism} from a complete Elgot $\hash$-algebra $(A,a,\argument^\istar)$ to a complete
  Elgot $\hash$-algebra $(B,b,\argument^\iistar)$ is a $\BC$-morphism $f:A\to B$, such that
  $((\IB{f}{\id})\comp e)^\iistar=f\comp e^\istar$ for all $e:X\to \IB{A}{X}$.
  This defines the category of complete $\IB{}{}$-algebras $\cbElg(\BC)$.
\end{definition}
\begin{remark}\label{rem:classicElgotalg}
  Note that complete Elgot $\IB{}{}$-algebras for the parametrized monad $\IB A X = A + \Sigma X$ (i.e.~the parametrized monad of Example~\ref{ex:pmons} (i) for $\BBT$ the identity monad) are precisely the complete Elgot algebras for the functor $\Sigma$ introduced and studied in~\cite{AdamekMiliusEtAl06j}. 
\end{remark}
\iffull
\begin{expl}\label{ex:simple2}
  Let us come back to the three simple \pmonad\/s on $\Set$ in
  Example~\ref{ex:simple} whose algebras are in each case simply
  binary algebras.  In each of the three cases, morphisms
  $X \to \IB{A}{X}$ can be understood as a systems of mutual recursive
  equations of a certain type with variables from the set $X$, and the $\argument^\istar$ operation of
  a complete Elgot $\IB{}{}$-algebra provides a solution of a given
  system of equations. However, the type of these systems of recursive
  equations is different for each of the three \pmonad\/s. For $A \mathbin{\#_1} X = A + X \times X$, $e: X \to A + X \times X$ can be
  understood as specifying for every $x \in X$ precisely one equation of one of the two types below:
  \[
    x \approx x' * x''\quad\text{with $x', x'' \in X$}\qquad\text{or}\qquad x \approx a\quad\text{with $a \in A$}.
  \]
  The solution $e^\istar$ then provides for every $x \in X$ an element
  $x^\istar \in A$ that, when plugged into the above formal equations,
  turn them into identities in $A$ where $*$ is interpreted as the
  binary operation of $A$.

  For $A \mathbin{\#_2} X = A \times X^*$, to give a morphism $e: X \to A \times X^*$ is equivalent to give a system of recursive equations 
  which specifies for each $x \in X$ an equation
  \[
    x \approx a * x'\quad\text{with $x' \in X$, $a\in A$},
  \]
  i.e.~iteration is restricted to the second argument.

  Finally, for $A \mathbin{\#_2} X = A \times X^*$, morphisms
  $e: X \to BA$ simply specify for each $x \in X$ a binary tree
  $e(x)$, and by the solution axiom, $e^\istar(x)$ is then the
  interpretation of this binary tree in $A$. Thus, iteration is
  trivial, in other words, every binary algebra is a complete Elgot
  algebra for $\#_3$.
\end{expl}
\fi
\begin{expl}\label{ex:cpoalg}
  Continuous algebras are complete Elgot $\IB{}{}$-algebras. Consider
  any category $\BC$ that is enriched over $\Cppo$ such that composition is left strict and a \pmonad\
  $\IB{}{}$ that is \emph{locally continuous} in both arguments,
  i.e.~$\bigsqcup_i (\IB{f_i}{g_i}) = \IB{(\bigsqcup_i
    f_i)}{(\bigsqcup_i g_i)}$ holds for any $f_i: A \to B$ and
  $g_i: X \to Y$. Then every $\IB{}{}$-algebra becomes a complete
  Elgot $\IB{}{}$-algebra when equipped with the operation
  $\argument^\istar$ assigning to every $e: X \to \IB{A}{X}$ its least
  solution. In more detail, let $A$ be a $\IB{}{}$-algebra, to every $e: X \to \IB{A}{X}$ we assign $e^\istar: X \to A$ given by 
  \[
    e^\istar = \bigsqcup\nolimits_i e_i^\istar, 
  \]
  where $e_0^\istar = \bot: X \to A$ and
  $e_{i+1} = a \comp (\IB{\id}{e_i^\istar}) \comp e$. That means that
  $e^\istar$ is the least fixed point of the function
  $s \mapsto a \comp (\IB{\id}{s}) \comp e$ on $\Hom_\BC(X,A)$. The verification that this satisfies the axioms of a complete Elgot $\IB{}{}$-algebra can be found \iffull\/in the appendix.\else\/in the full version of our paper.\fi
\end{expl} 
Note that we did not require a morphism of complete Elgot $\hash$-algebras 
to be a morphism of $\hash$-algebras. Somewhat surprisingly, this follows 
automatically.
\begin{prop}\label{prop:hom}
Let $f:A\to B$ be a complete Elgot $\hash$-algebra morphism from $(A,a, \argument^\istar)$ to $(B,b, \argument^\iistar)$.
Then f is a morphism of $\hash$-algebras.
\end{prop}
\begin{proof*}{Sketch}
The idea is to represent $a$ as a loop terminating after the first iteration and 
then deduce preservation of $a$ by $f$ from preservation of iteration by 
$f$ guaranteed by definition. More concretely, we take  
\[e = (\IB{\id}{\inr}) \comp [\id, u_A^A] : (\IB{A}{A}) + A \to
    \IB{A}{((\IB{A}{A})+A)}\]
and show that $e^{\istar}=[a, \id]$. The remaining proof amounts to deriving $b\comp (\IB{f}{f})=f\comp a$
 from $f\comp e^\istar = ((\IB{f}{\id})\comp e)^\iistar$.
\qed\end{proof*}
It was shown by Uustalu~\cite{Uustalu03}
that parametrized monads give rise to monads at least in two different ways:
\begin{prop}\label{prop:tarmo}
Suppose, $\IB{}{}$ is a parametrized monad on $\BC$ such that the least fixpoint 
$\mu\gamma.\ \IB{X}{\gamma}$ (the greatest fixpoint $\nu\gamma.\ \IB{X}{\gamma}$) 
exists for every $X\in |\BC|$. Then $\mu\gamma.\IB{\argument}{\gamma}$ 
($\nu\gamma.\IB{\argument}{\gamma}$) is the underlying functor of a monad.
\end{prop}
It is known that the initial algebra $\mu\gamma.\ \IB{X}{\gamma}$ yields the free $\IB{}{}$-algebra on~$X$; in fact, existence of this free $\IB{}{}$-algebra is equivalent to the existence of that initial algebra (see~\cite[Theorem~2.18]{JirAdamekMiliusEtAl08}). Here we are interested in the final coalgebras $\nu\gamma.\,\IB{X}{\gamma}$. These yield the free complete $\IB{}{}$-algebras, and moreover, existence of these free algebras is equivalent to the existence of that final coalgebras.
\begin{thm}\label{thm:eq}
  \begin{enumerate}
  \item Suppose that $\out_X: \IBnu X \to \IB{X}{\IBnu X}$ is a final $(\IB{X}{\argument})$-coalgebra. Then the following morphisms
    \[
      \begin{tikzcd}
        \IB{\IBnu X}{\IBnu X} \ar[rr, "\IB{\out_X}{\id}"] && \IB{(\IB{X}{\IBnu X})}{\IBnu X} \ar[r, "m^{\IBnu X}_X"] & \IB{X}{\IBnu X} \ar[r, "{\tuo[X]}"] & \IBnu X  
      \end{tikzcd}
    \]
    and
    \[
      \begin{tikzcd}
        X \ar[r, "u^{\IBnu X}_X"] & \IB{X}{\IBnu X} \ar[r, "{\tuo[X]}"] & \IBnu X 
      \end{tikzcd}
    \]
    form the algebra structure and universal morphism of a free complete Elgot algebra for $\IB{}{}$ on $X$. 
  \item Suppose that $\varphi_X: \IB{FX}{FX} \to FX$ and $\eta_X: X \to FX$ form a free complete Elgot $\IB{}{}$-algebra  on $X$. Then
    \[
      \begin{tikzcd}
        \IB{X}{FX} \ar[rr, "\IB{\eta_X}{\id}"] && \IB{FX}{FX} \ar[r, "\varphi_X"] & FX  
      \end{tikzcd}
    \]
    is an isomorphism, and its inverse is the structure of a final ($\IB{X}{\argument}$)-coalgebra.
  \end{enumerate}
\end{thm}
\begin{remark}\label{rem:it}
  Note that in Clause~(i) above the iteration operator on $\IBnu Y$ is obtained as follows. Given $e: X \to \IB{\IBnu Y}{X}$ one forms the following coalgebra $c: \IBnu Y + X \to \IB{Y}{(\IBnu Y + X)}$ for $\IB{Y}{\argument}$:
  \[
    \begin{tikzcd}
      \IBnu Y + X 
      \ar[r, "{[u^X_{\IBnu Y}, e]}"]
      &
      \IB{\IBnu Y}{X}
      \ar[r, "\IB{\out}{\id}"]
      &
      \IB{(\IB{Y}{\IBnu Y})}{X}
      \ar[d, "\IB{(\IB{\id}{\inl})}{\inr}"]
      \\
      \IB{Y}{(\IBnu Y +X)}
      &
      &
      \IB{(\IB{Y}{(\IBnu Y + X)})}{(\IBnu Y + X)}
      \ar[ll, "m^{\IBnu Y + X}_Y"']
    \end{tikzcd}
  \]
  Then one puts $e^\dagger = (\coit c)\comp \inr$. 
\end{remark}
The proof of Theorem~\ref{thm:eq} is a non-trivial generalization of the proof
of~\cite[Theorem~5.4]{AdamekMiliusEtAl06j} from complete Elgot algebras for
endofunctors to those for {\pmonad}s; we will establish Clause~(i)~as a consequence of
Theorem~\ref{thm:algebras-iso} (see Corollary~\ref{cor:algebras-iso}) while we
outline the proof of Clause~(ii)~in the \iffull\/appendix.\else\/full version of the paper.\fi

Before we continue, let us note that, surprisingly, in a free complete Elgot algebra
the iteration always assigns a unique solution to any
$e: X \to \IB {FX} X$.
\begin{prop}\label{prop:unique}
  Suppose that $\varphi_Y: \IB{FY}{FY} \to FY$ and $\eta_Y: Y \to FY$ form a free complete Elgot $\IB{}{}$-algebra on $Y$. Then for every $e: X \to \IB {FY} X$, $e^\dagger: X \to FY$ is a unique solution, i.e.\ a unique morphism satisfying the
solution axiom with $e$.
\end{prop}
From now on we assume that the final coalgebras $\nu\gamma.\,\IB{X}{\gamma}$ exist and denote them $\IBnu X$ (standardly omitting the structure morphisms $\out_X : \IBnu X \to \IB{X}{\IBnu X}$). Recall that $\coit f:X\to \IBnu Y$ is the morphism 
uniquely induced by a coalgebra $(X,f:X\to \IB{Y}{X})$. Following~\cite{Uustalu03}, in order to introduce and reason about the monad structure of $\IBnu$, we use a more flexible \textit{primitive corecursion principle}, derived from the standard \emph{coiteration principle} embodied in $\coit$.
\begin{prop}[\cite{Uustalu03}]\label{prop:corec}
For any functor $F$ with a final coalgebra $\nu F$, and any $f :
X \to F(\nu F + X)$, there is a unique morphism $h$ satisfying $\out \comp h =
F[\id, h] \comp f$.
\end{prop}
The morphism $h$ in Proposition~\ref{prop:corec} is said to be defined by 
primitive corecursion. We use primitive corecursion to slightly generalize the
$\coit$ construct in the special case of $\IBnu$: 
\begin{lem}\label{lem:IBcorec}
For any $e : X \to \IB{B}{X}$ and $f : B
\to \IB{A}{\IBnu{A}}$, there is a unique morphism $h$ satisfying
\begin{equation}\label{eq:corec}
  \begin{tikzcd}
    X
    \arrow[d, "h"']
    \arrow[r, "e"]
    & \IB{B}{X}
    \arrow[d, "{m_A^{\IBnu{A}} \comp (\IB{f}{h})}"] \\
    \IBnu{A}
    \arrow[r, "\out"]
    & \IB{A}{\IBnu{A}}.
  \end{tikzcd}
\end{equation}
\end{lem}
For any $e : X \to \IB{B}{X}$ and $f : B \to \IB{A}{\IBnu{A}}$ we
denote by
\[
  \corec(e, f) : X \longrightarrow \IBnu{A}
\]
the unique $h$ making diagram~\eqref{eq:corec} commute.
Using \eqref{eq:corec}, the monad structure on $\IBnu$ can be given as follows:\smnote{From here on the subscripts of $\out$, $u$, $\etan$ etc. are usually dropped. This is very confusing to the reader who must now guess domain/codomain -- e.g.~in the proof of 4.15 one is lost; in any case it should not be done without telling the reader!}
\begin{align*}
  \etan_X &= \tuo \comp u_X^{\IBnu X} = \coit u_X^X \\
  f^{\klstar} &= \corec((\IB{f}{\id}) \comp \out, \out) & \text{where } f : X \to \IBnu Y
\end{align*}
This also defines $\mun=\id^\klstar=\corec(\out, \out)$. Note that, by 
Lemma~\ref{lem:IBcorec}, $f^\klstar$ is the unique morphism satisfying equation
\begin{align}\label{eq:klstar_unique}
\out f^\klstar = m_Y^{\IBnu Y}\comp (\IB{\out f}{f^\klstar})\comp\out.
\end{align}
\vspace{-3ex}
\begin{lem}\label{lem:corec-coit}
  Let $e : X \to \IB{B}{X}$ and $f : B \to \IB{A}{\IBnu{A}}$. Then
  \[ \corec(e, f) = (\tuo f)^{\klstar} \comp (\coit e). \]
\end{lem}
As an easy corollary of Lemma~\ref{lem:corec-coit} we obtain that 
$\coit e =\corec(e,u_X^{\IBnu X})$; indeed, we have
\[
\corec(e,u_X^{\IBnu X}) = (\tuo \comp u_X^{\IBnu X})^{\klstar} \comp (\coit e) = (\eta_X^\nu)^{\klstar} \comp (\coit e) = \coit e.
\]
We state another useful property in the following lemma:
\begin{lem}\label{lem:corec-nat}
Let $e : X \to \IB{B}{X}$ and $g : B \to C$. Then
  \[ \IBnu g \comp (\coit e) = \coit((\IB{g}{\id}) \comp e). \]
\end{lem}
The following theorem is our first main result. It establishes an equivalence of
complete Elgot $\IB{}{}$-algebras and $\IBnu$-algebras.
\begin{thm}\label{thm:algebras-iso}
  For any {\pmonad}\/ $\IB{}{}: \BC \times \BC \to \BC$, the Eilenberg-Moore
  algebras of\/ $\IBnu = \nu\gamma.\,\IB{\argument}{\gamma}$ are precisely the
  complete Elgot $\IB{}{}$-algebras. More precisely,\/ 
  $\BC^{\IBnu}$ and $\cbElg(\BC)$ are isomorphic categories under the identical on morphisms 
isomorphism constructed as follows:
\begin{itemize}
 \item $\BC^{\IBnu}\to\cbElg(\BC)$: for a $\IBnu$-algebra $(A,\chi:\IBnu A\to A)$
we define a $\IB{}{}$-algebra $(A,\chi\tuo\comp(\IB{\id}{\eta^\nu}):\IB{A}{A}\to A,\argument^\istar)$ with
$e^\istar=\chi\comp(\coit e):X\to A$ for any $e:X\to\IB{A}{X}$. 
 \item $\cbElg(\BC)\to\BC^{\IBnu}$: for a $\IB{}{}$-algebra 
$(A,a:\IB{A}{A}\to A,\argument^\istar)$ we define a $\IBnu$-algebra $(A,\out^\istar:\IBnu A\to A)$.
\end{itemize}
\end{thm}
\begin{proof*}{Sketch} 
For the direction from $\BC^{\IBnu}$ to $\cbElg(\BC)$ we have to verify the axioms of
complete Elgot $\IB{}{}$-algebras. The hardest case is that of the compositionality
identity.   We have on the one hand
  \begin{flalign*}
    &&(f^{\istar} \bullet g)^{\istar} =&~ \chi \comp \coit(f^{\istar} \bullet g) \\
    &&=&~\chi \comp \coit\left((\IB{\chi \comp (\coit f)}{\id}) \comp g\right) \\
    &&=&~\chi \comp \coit\left((\IB{\chi}{\id}) \comp (\IB{(\coit f)}{\id}) \comp g\right) \\
    &&=&~\chi \comp (\IBnu{\chi}) \comp \coit((\IB{(\coit f)}{\id}) \comp g) & \by{Lemma~\ref{lem:corec-nat}} \\
    &&=&~\chi \comp \mun  \coit((\IB{(\coit f)}{\id}) \comp g)  & \by{$\chi$ is an $\IBnu$-algebra} \\
    &&=&~\chi \comp \corec((\IB{(\coit f)}{\id}) \comp g, \out), & \by{Lemma~\ref{lem:corec-coit}}
\intertext{and on the other hand, by definition,
}
    && (f \filledsquare g)^{\istar} \comp \inr =&~
\chi \comp \coit(m^{Y+X}_A \comp (\IB{((\IB{\id}{\inl}) \comp f)}{\inr}) \comp [u_Y^X, g]) \comp \inr.&
\end{flalign*} 
Let us denote $m^{Y+X}_A \comp (\IB{((\IB{\id}{\inl}) \comp f)}{\inr}) \comp
  [u_Y^X, g]$ by $h$. By Lemma~\ref{lem:IBcorec}, it suffices to show the identity
$
\out \comp (\coit h) \comp \inr = m^{\IBnu{A}}_A \comp (\IB{\out}{((\coit h) \comp \inr)}) \comp (\IB{\coit f}{\id}) \comp g.
$ 
The latter is easy to obtain from the auxiliary equation $(\coit h)\comp\inl=\coit f$ whose proof is a routine.

For the direction from $\cbElg(\BC)$ to $\BC^{\IBnu}$, we have to prove the two
axioms of Eilenberg-Moore algebras. The harder one is 
$\out^\istar\comp\IBnu(\out^\istar)=\out^\istar\comp\mu^\nu$ and it is obtained 
from the instance of compositionality 
$(\out \filledsquare \out)^{\istar}\inr=(\out^{\istar} \bullet \out)^{\istar}$ 
by establishing
$\out^{\istar} \comp [\id, \mun] = (\out \filledsquare \out)^{\istar}$ and
$ \out^{\istar} \comp \IBnu{(\out^{\istar})} = (\out^{\istar} \bullet \out)^{\istar}$.
Further calculations ensure that the correspondence between $\cbElg(\BC)$ and
$\BC^{\IBnu}$ is functional and moreover an isomorphism.
\qed\end{proof*}
\begin{cor}\label{cor:algebras-iso}
  Free complete Elgot $\IB{}{}$-algebras exist for all objects $A$ of\/ $\BC$ if and only if
  the final coalgebras $\IBnu{A}$ exist. The functor assigning to an object
  $A$ its free complete Elgot $\IB{}{}$-algebra is
  $
    FA = \bigl( \IBnu{A}, \tuo[A] \comp m_A^{\IBnu{A}} \comp
      (\IB{\out_A}{\id_{\IBnu{A}}}), \argument^\dagger \bigr),
  $
  where $\argument^\dagger$ is the iteration operation defined in Remark~\ref{rem:it}.
\end{cor}

\section{Algebras of Complete Elgot Monads}\label{sec:alg}
We are now in a position to apply the results on complete Elgot $\IB{}{}$-algebras
developed in the previous section to explore the connection between complete Elgot monads
and complete Elgot algebras. We briefly motivate our further technical contribution as follows.

Recall that given a monad $\BBT$ and an endofunctor $\Sigma$ over $\BC$, 
$\IB{X}{Y}=T(X+\Sigma Y)$ is a parametrized monad and therefore, by Proposition~\ref{prop:tarmo},
$\BBT_{\Sigma}$ given by~\eqref{eq:TF} is a monad. We reserve notation $\BBT_{\nu}$
for the special case when $\Sigma=\Id$:
\begin{align*}
T_{\nu} X =\nu\gamma.\, T(X+\gamma).
\end{align*}
From a computational point of view, $T_{\nu} X$ can be considered as a type of
processes triggering a computational effect formalized by $\BBT$ at each step and
eventually outputting values from $X$ in case of successful termination.
The unary operation captured by $\Sigma=\Id$ intuitively means the action of
\emph{delaying}. This perspective was previously pursued in~\cite{GoncharovSchroder13}.
Now, if $\BBT$ is a complete Elgot monad, or more generally, any monad equipped
with an iteration operator, we can define a \emph{collapsing morphism} 
$\delta_X:T_{\nu}X\to TX$ as follows:
\begin{align}\label{eq:collapse}
\delta_X = \Bigl(T_{\nu} X\xrightarrow{\out_X} T(X+T_{\nu}X) \Bigr)^\istar,
\end{align}
which intuitively flattens every possibly infinite sequence of computational steps
of $T_{\nu} X$ into a single step of $TX$. Let us illustrate this with the following
toy example.
\begin{expl}\label{expl:toy}
Let $TX=\PSet_{\omega_1}(A^\star\times X)$ where $\PSet_{\omega_1}$ is the countable
powerset functor and $A$ is some fixed alphabet of \emph{actions} like in Example~\ref{expl:proc}. 
We extend\/ $T$ to a monad~$\BBT$ by putting 
\begin{align*}
\eta_X(x)=\{(\varepsilon,x)\}&&\text{and}&& f^\klstar(s\subseteq A^\star\times X) = \{(ww',y)\mid (w,x)\in s, (w',y)\in f(x)\},
\end{align*}
where $\varepsilon\in A^*$ is the empty word and $f:X\to \PSet_{\omega1}(A^\star\times Y)$. It is easy
to see that~$\BBT$ is an $\omega$-continuous monad (see Example~\ref{def:omega-cont}) and hence
a complete Elgot monad with the iteration operator defined using least fixed points. 
An element of $TX$ is intuitively a countably branching process,
with results in $X$, at each step capable of executing a finite series of actions.
Now the collapsing morphism~\eqref{eq:collapse} for every process $p\in T_{\nu}\{\checkmark\}$ 
calculates the set $\mathsf{tr}(p)\subseteq A^\star$ of all sucessful traces of $p$. 
\end{expl}
As we will see latter (Theorem~\ref{thm:emon}~(i)), 
$
\bigl(TX,~T_\nu TX\xrightarrow{\delta_{TX}} TTX\xrightarrow{\mu_X} TX\bigr)
$ 
is a $\BBT_{\nu}$-algebra
and hence, by Theorem~\ref{thm:algebras-iso}, a complete Elgot $\IB{}{}$-algebra. 
We can now change the perspective and instead of $TX$ consider an arbitrary 
complete Elgot $\IB{}{}$-algebra. The question we consider next is: Is it possible
to recover the laws of iteration for $\BBT$ assuming that every $\BBT$ is coherently
equipped with the structure of a complete Elgot $\IB{}{}$-algebra? It turns out that
without any further assumptions on the category of complete Elgot $\IB{}{}$-algebras
almost all laws of complete Elgot monads become derivable. More precisely, we introduce
the following class of monads.
\begin{defn}
A monad\/ $\BBT$ is called a \emph{weak complete Elgot monad} if it is equipped with an
iteration operator $\argument^\istar$ that satisfies \emph{fixpoint}, \emph{naturality},
and \emph{uniformity} axioms and the following identity: for any $g:X\kto Y + X$, $f:Y\kto Z+Y$ we have
\begin{align}\label{eq:waterfall2}
  \Bigl(Y+X \xrightarrow{[\inl, g]} Y + X \xrightarrow{f +\id} Z + Y + X\Bigr)^\istar \inr = X \xrightarrow{g^\istar} Y \xrightarrow{f^\istar} Z.
\end{align}
(See Fig.~\ref{fig:ax_weak} for the pictorial form.)%
\end{defn}
It is relatively easy to deduce~\eqref{eq:waterfall2} from the codiagonal identity, hence
we obtain
\begin{figure}[t!]
\tikzset{
    font=\tiny,
    nonterminal/.style={
      rectangle,
      minimum size=6mm,
      very thick,
      draw=orange!50!black!50,         
      fill=orange!50!white,
      font=\itshape
    },
    terminal/.style={
      scale=.5,
      circle,
      inner sep=0pt,
      thin,draw=black!50,
      top color=white,bottom color=black!20,
      font=\ttfamily
    },
    iterated/.style={
      fill=green!20,
      thick,
      draw=green!50
    },
    natural/.style={
      circle,
      minimum size=4mm,
      inner sep=2pt,
      thin,draw=black!50,
      top color=white,bottom color=black!20,
      font=\ttfamily},
    skip loop/.style={to path={-- ++(0,#1) -| (\tikztotarget)}}
  }

  {
    \tikzset{nonterminal/.append style={text height=1.5ex,text depth=.25ex}}
    \tikzset{natural/.append style={text height=1.5ex,text depth=.25ex}}
  }
  \captionsetup[subfigure]{labelformat=empty,justification=justified,singlelinecheck=false}
  \pgfdeclarelayer{background}
  \pgfdeclarelayer{foreground}
  \pgfsetlayers{background,main,foreground}

    \centering
    \raisebox{-.5\height}{
    \begin{tikzpicture}[
      point/.style={coordinate},>=stealth',thick,draw=black!50,
      tip/.style={->,shorten >=0.007pt},every join/.style={rounded corners},
      hv path/.style={to path={-| (\tikztotarget)}},
      vh path/.style={to path={|- (\tikztotarget)}},
      text height=1.5ex,text depth=.25ex 
      ]
      \node [nonterminal] (f) {$g$};
      \node [nonterminal] (g) at ($(f.east)+(1.5,.75)$) {$f$};
      \draw [<-] (f.west) -- +(-1,0) node [midway,above] {$X$};
      \path [<-,draw] (f.west)++(-0.5,0) 
        -- ++(0,-0.75) 
        -| ($(f.east)+(2.5,-.15)$) node [near end,right] {$X$} 
        -- ($(f.east)+(0,-.15)$);
      \draw [->] ($(f.east)+(0,0.15)$) -| ($(f.east)+(0.45,0.15)$) node [pos=0.2,above] {$Y$}
        -- ($(f.east)+(0.45,.75)$);
      \draw [->] ($(g.east)+(0,.15)$) -- +(1,0) node [midway,above] {$Z$};
      \draw [->] ($(g.east)-(0,.15)$) -- +(1,0) node [midway,below] {$Y$};
      \draw [<-] (g.west) -- +(-2.8,0) node [pos=0.82,above] {$Y$};
      \begin{pgfonlayer}{background}
        \draw [iterated] ($(f.north west)+(-0.25,-.85)$) rectangle ($(g.south east)+(0.25,0.85)$);
      \end{pgfonlayer}
    \end{tikzpicture}
}~~=~~
    \raisebox{-.5\height}{
    \begin{tikzpicture}[
      point/.style={coordinate},>=stealth',thick,draw=black!50,
      tip/.style={->,shorten >=0.007pt},every join/.style={rounded corners},
      hv path/.style={to path={-| (\tikztotarget)}},
      vh path/.style={to path={|- (\tikztotarget)}},
      text height=1.5ex,text depth=.25ex 
      ]
      \node [nonterminal] (f) {$g$};
      \node [nonterminal] (g) at ($(f.east)+(2.5,0.15)$) {$f$};
      \draw [<-] (f.west) -- +(-1,0) node [midway,above] {$X$};
      \path [<-,draw] (f.west)++(-0.5,0) -- ++(0,-0.8) -| ($(f.east)+(0.5,-0.15)$) node [near end,right] {$X$}  -- ++(-0.5,0);
      \draw [->] ($(f.east)+(0,0.15)$) -- (g) node [pos=0.25,above] {$Y$};
      \draw [->] ($(g.east)+(0,.15)$) -- ++(1,0) node [midway,above] {$Z$};
      \path [<-,draw] (g.west)++(-0.5,0) -- ++(0,-0.8) -| ($(g.east)+(0.5,-0.15)$) node [near end,right] {$Y$}  -- ++(-0.5,0);
      \begin{pgfonlayer}{background}
        \draw [iterated] ($(f.north west)+(-0.25,0.25)$) rectangle ($(f.south east)+(0.25,-0.25)$);
        \draw [iterated] ($(g.north west)+(-0.25,0.25)$) rectangle ($(g.south east)+(0.25,-0.25)$);
      \end{pgfonlayer}
    \end{tikzpicture}
}
\vspace{1ex}
\caption{The additional axiom for weak complete Elgot monads.}
\label{fig:ax_weak}
\end{figure}

\begin{prop}\label{prop:elgot_to_weak}
Any complete Elgot monad is a weak complete Elgot monad.
\end{prop}  
We now can establish a tight connection between weak complete Elgot monads and
complete Elgot $\IB{}{}$-algebras.
\begin{thm}\label{thm:emon}
Let $\BBT$ be a monad on $\BC$ and let $\IB{X}{Y} = T(X + Y)$.
\begin{enumerate}
 \item If\/ $\BBT = (T, \eta, \argument^\klstar, \argument^\istar)$ is a weak 
complete Elgot monad then $\BC^{\BBT}$ is isomorphic to
the full subcategory of $\ \cbElg(\BC)$ formed by those complete Elgot $\IB{}{}$\dash algebras 
$(A,a:T(A+A)\to A,\argument^\iistar)$ which factor through $T\nabla:T(A+A)\to TA$ and
for which $e^\iistar = a\comp (T\inl)\comp e^\istar$ for every $e:X\to T(A+X)$.
 \item Conversely, any functor $J:\BC^{\BBT}\to \cbElg(\BC)$ sending a $\BBT$-algebra
\mbox{$a:TA\to A$} to $a \comp (T\nabla):T(A+A)\to A$ and identical on morphisms
induces a weak complete Elgot monad structure on $\BBT$ as follows:
\begin{align}\label{eq:iistar_def}
\frac{e:X\to T(Y+X)}{e^\istar = (T(\eta+\id)\comp e)^\iistar:X\to TY}
\end{align}
where $\argument^{\iistar}$ is the iteration operator on $J(TY,\mu)$ (by Clause~(i), 
$J$ is then full and faithful). 
\end{enumerate}
\end{thm}
If $\cbElg(\BC)$ additionally satisfy a version of the codiagonal identity, the 
construction from Clause~(ii) of Theorem~\ref{thm:emon} produces precisely complete
Elgot monads.
\begin{thm}\label{thm:elgot_from_alg}
Let $\BBT$ be a monad on $\BC$, let $\IB{X}{Y} = T(X + Y)$ and let 
$J:\BC^{\BBT}\to \cbElg(\BC)$ be a functor as in Clause~(ii) of Theorem~\ref{thm:emon}.
Then $\BBT$ is equipped with the structure of a weak complete Elgot monad given by~\eqref{eq:iistar_def}, 
and moreover $\BBT$ is a complete Elgot monad iff every $(A,a,\argument^\iistar)$ in $\cbElg(\BC)$
satisfies the equation
\begin{align}\label{eq:codiag_alg}
(m^X_{\IB{A}{X}}\comp e)^\iistar = (e^\iistar)^\iistar
\end{align} 
for every $e:X\to \IB{(\IB{A}{X})}{X}$ (this uses the fact that $\IB{A}{X}=T(A+X)$ is a 
free $\BBT$-algebra and hence a complete Elgot $\IB{}{}$-algebra).
\end{thm}

\section{Conclusions and Further Work}
We introduced the notion of complete Elgot algebra for a
parametrized monad, based on the previous work~\cite{Uustalu03,AdamekMiliusEtAl05}.
We showed that the category of complete Elgot algebras for a parametrized monad $\IB{}{}$ is
isomorphic to the category of Eilenberg-Moore algebras for the monad $\nu\gamma.\,\IB{\argument}{\gamma}$
whenever the latter exists. As the category of complete Elgot $\IB{}{}$-algebras is
given axiomatically, this can be considered as a form of soundness and completeness
result, specifically, it indicates that algebras for
$\nu\gamma.\,\IB{\argument}{\gamma}$ are subject to a lightweight theory of (uniform) 
iteration. 

We explored the connection between complete Elgot $\IB{}{}$-algebras 
for $\IB{X}{Y}=T(X+Y)$ and Eilenberg-Moore algebras of complete Elgot monads, i.e.\ monads
from~\cite{GoncharovRauchEtAl15} supporting a uniform iteration operator satisfying standard 
axioms of iteration.
Specifically, we showed that monads $\BBT$ whose algebras are coherently equipped with the structure
of a complete Elgot $\IB{}{}$-algebra are precisely complete Elgot monads with the 
codiagonal axiom replaced by its weakened form (Theorem~\ref{thm:emon}). Moreover,
if the category of complete Elgot $\IB{}{}$-algebras satisfies a variant of the 
codiagonal law, such monads $\BBT$ are complete Elgot monads (Theorem~\ref{thm:elgot_from_alg}).

As further work we plan to improve Theorem~\ref{thm:elgot_from_alg} to
obtain an intrinsic characterization of complete Elgot monads in the
style of Theorem~\ref{thm:emon} (i.e.~without assuming extra
properties of the complete Elgot algebras). We believe that the
results we obtained are potentially useful for facilitating
constructions over complete Elgot monads, in particular we a seeking
for a conceptual simplification for the sophisticated proofs
underlying the main result of~\cite{GoncharovRauchEtAl15} stating
that~\eqref{eq:TF} is a complete Elgot monad whenever $\BBT$ is. Also
we are interested in applications of the obtained results to semantics
of abstract side-effecting processes in the style
of~\cite{GoncharovSchroder13}.

\bibliographystyle{entcs}
\bibliography{monads}

\iffull 
\clearpage
\appendix
\allowdisplaybreaks

\section{Appendix: Omitted proofs}

\subsection*{Proof of Proposition~\ref{prop:bekic-from-elgot} (Beki\'{c} identity)}
  Let us show that complete Elgot monads validate the Beki\'{c}
  identity. Let
  \begin{align*}
    u = T((\id + \inl) + \inr) \comp [f, g] : Y + X \to T((Z+(Y+X))
    + (Y+X)).
  \end{align*}
  By \textit{codiagonal},
  \begin{equation}
    (T[\id, \inr] \comp u)^{\istar} =
    (u^{\istar})^{\istar}.\label{eq:u-codiag}
  \end{equation}
  Now the left-hand side of \eqref{eq:u-codiag} simplifies to
  \begin{align*}
    &~(T[\id, \inr] \comp T((\id + \inl) + \inr) \comp [f,g])^{\istar} \\
    =&~(T[\id + \inl, \inr \comp \inr] \comp [f, g])^{\istar} \\
    =&~(T\alpha \comp [f, g])^{\istar},
  \end{align*}  
  i.e.~to the left-hand side of the Beki\'{c} identity.  Now observe
  that, by \textit{uniformity} and \textit{naturality},
  \begin{align}\label{eq:bek_proof}
    u^{\istar} \comp \inr = (T(\id + \inl) + \id) \comp g)^{\istar} = T(\id
    + \inl) \comp g^{\istar}.
  \end{align}
  Therefore, the right-hand side of \eqref{eq:u-codiag} can be
  rewritten in the form
  \begin{flalign*}
    &&(u^{\istar})^{\istar} =&~ ([\eta, u^{\istar}]^{\klstar} \comp u)^{\istar}&\by{fixpoint} \\
    &&=&~ ([\eta \comp (\id + \inl), u^{\istar} \comp \inr] \comp [f, g])^{\istar} \\
    &&=&~ ([T(\id + \inl) \comp \eta, T(\id + \inl) \comp g^{\istar}]^{\klstar} \comp [f, g])^{\istar} &\by{\eqref{eq:bek_proof}}\\
    &&=&~ (T(\id + \inl) \comp [\eta, g^{\istar}]^{\klstar} \comp [f, g])^{\istar} \\
    &&=&~ (T(\id + \inl) \comp [[\eta, g^{\istar}]^{\klstar} \comp f, g^{\istar}])^{\istar} &\by{fixpoint} \\
    &&=&~ ([\eta \comp \inl, \eta \comp \inr \comp \inl]^{\klstar} \comp [[\eta, g^{\istar}]^{\klstar} \comp f, g^{\istar}])^{\istar} \\
    &&=&~ [\eta, ([\eta \comp \inl, [[\eta, g^{\istar}]^{\klstar} \comp f, g^{\istar}]]^{\klstar} \comp \eta \comp \inr \comp \inl)^{\istar}]^{\klstar}\\&&& ~~\comp [[\eta, g^{\istar}]^{\klstar} \comp f, g^{\istar}] &\by{dinaturality} \\
    &&=&~ [\eta, ([\eta, g^{\istar}]^{\klstar} \comp f)^{\istar}]^{\klstar} \comp [[\eta, g^{\istar}]^{\klstar} \comp f, g^{\istar}] \\
    &&=&~ [([\eta, g^{\istar}]^{\klstar} \comp f)^{\istar}, [\eta, ([\eta, g^{\istar}]^{\klstar} \comp f)^{\istar}]^{\klstar} \comp g^{\istar}] &\by{fixpoint} \\
    &&=&~ [h^{\istar}, [\eta, h^{\istar}]^{\klstar} \comp g^{\istar}] \\
    &&=&~ [\eta, h^{\istar}]^{\klstar} \comp [\eta \comp \inr, g^{\istar}],
  \end{flalign*}
  i.e.~equals the right-hand side of the Beki\'{c} identity.

  For the opposite direction, we need to show that the Beki\'{c}
  identity implies \textit{dinaturality} and \textit{codiagonal}. For
  the latter, let $k : X \to T((Y+X) + X)$. By the Beki\'{c} identity,
  \begin{align*}
    (T\alpha \comp [k,k])^{\istar} =&~ [\eta, ([\eta, k^{\istar}]^{\klstar} \comp
    k)^{\istar}]^{\klstar} \comp [\eta \comp \inr, k^{\istar}] \\
    =&~ [\eta, (k^{\istar})^{\istar}]^{\klstar} \comp [\eta \comp \inr, k^{\istar}].
\intertext{Thus, $(T\alpha \comp [k, k])^{\istar} \comp \inl = (T\alpha \comp [k,
  k])^{\istar} \comp \inr = (k^{\istar})^{\istar}$. On the other hand, by
  \textit{uniformity},}
  (T\alpha \comp [k, k])^{\istar} =&~ (T[\id, \inr] \comp k)^{\istar} \comp [\id, \id]
\intertext{and therefore}
  (k^{\istar})^{\istar} =&~ (T\alpha \comp [k, k])^{\istar} \comp \inr =
    (T[\id, \inr] \comp k)^{\istar}
\end{align*}
  as required.
  To prove \textit{dinaturality}, we define the term
  \[
    w = ([T(\id + \inr) \comp h, T(\id + \inl) \comp g])^{\istar}
  \]
  for $g : X \to T(Y+Z)$, $h: Z \to T(Y+X)$. By uniformity,
  \[ w \comp [\inr,\inl] = ([T(\id + \inr) \comp g, T(\id + \inl) \comp
    h])^{\istar}. \]
 The Beki\'{c} identity then gives us
  \begin{align*}
    w =&~ ([T(\id + \inr) \comp h, T(\id + \inl) \comp g])^{\istar} \\
    =&~ (T\alpha \comp [ T(\inl + \id) \comp h, T\inl \comp g])^{\istar} \\
    =&~ [\eta, ([\eta, (T\inl \comp g)^{\istar}]^{\klstar} \comp T(\inl + \id) \comp h)^{\istar}]^{\klstar}
       \comp [\eta \comp \inr, (T\inl \comp g)^{\istar}] \\
    =&~ [\eta, ([\eta \comp \inl, g]^{\klstar} \comp h)^{\istar}]^{\klstar} \comp [\eta \comp \inr, g]
\intertext{as well as}
    w \comp [\inr, \inl] =&~ [\eta, ([\eta \comp \inl, h]^{\klstar} \comp g)^{\istar}]^{\klstar} \comp [\eta \comp \inr, h].
  \end{align*}
  Therefore,
  \begin{flalign*}
    &&([\eta \comp \inl, h]^{\klstar} \comp g)^{\istar} = w \comp [\inr, \inl] \comp \inl 
    = w \comp \inr = [\eta, ([\eta \comp \inl, g]^{\klstar} \comp h)^{\istar}]^{\klstar} \comp g.&&\text{\qed}
  \end{flalign*}

\subsection*{Proof of Proposition~\ref{prop:T-Sigma-bialg}}
  We define constructions to convert from $\IB{}{}$-algebras to
  $\BBT$-$\Sigma$-bialgebras and vice versa.
\begin{enumerate}
\item Given a $\hash$-algebra $\alpha : T(A+\Sigma A) \to A$, let
  \begin{align*}
    a =&~ TA \xrightarrow{~T\inl~} T(A + \Sigma A) \xrightarrow{~~\alpha~~} A \\
    f =&~ \Sigma A \xrightarrow{~\eta \comp \inr~} T(A + \Sigma A) \xrightarrow{~~\alpha~~} A.
  \end{align*}
  We immediately check that $a$ satisfies the properties of a $\BBT$-algebra:

\vspace{-3ex}
  \begin{minipage}[t]{.4\linewidth}
  \begin{align*}
    a \comp \eta =&~ \alpha \comp (T\inl) \comp \eta\\
    =&~ \alpha \comp \eta \comp \inl \\
    =&~ \alpha \comp u\\
    =&~ \id
  \end{align*}
  \end{minipage}
  \begin{minipage}[t]{.6\linewidth}
  \begin{align*}
    a \comp \mu =&~ \alpha \comp (T\inl) \comp \mu \\
    =&~ \alpha \comp \mu \comp (TT\inl) \\
    =&~ \alpha \comp \mu \comp T[\id,\eta \comp \inr] \comp (T\inl) \comp (TT\inl) \\
    =&~ \alpha \comp m \comp (T\inl) \comp (TT\inl) \\
    =&~ \alpha \comp T(\alpha + \id) \comp (T\inl) \comp (TT\inl) \\
    =&~ \alpha \comp (T\inl) \comp (T\alpha) \comp \comp (TT\inl) \\
    =&~ a \comp (Ta)\\[-2ex]
  \end{align*}
  \end{minipage}
\item Conversely, given a bialgebra $TA \xrightarrow{~~a~~} A \xleftarrow{~~f~~} \Sigma
  A$, form
  \[ \alpha = T(A + \Sigma A) \xrightarrow{~~T[\id, f]~~} TA \xrightarrow{~~a~~} A. \]
  The constructed $\alpha$ is a $\hash$-algebra:

\vspace{-3ex}
  \begin{minipage}[t]{.4\linewidth}
  \begin{align*}
    \alpha \comp u =&~ a \comp T[\id, f] \comp \eta \comp \inl\\
    =&~ a \comp \eta \\
    =&~ \id
  \end{align*}
  \end{minipage}
  \begin{minipage}[t]{.6\linewidth}
  \begin{align*}
    \alpha \comp m =&~ a \comp T[\id, f] \mu \comp T[\id, \eta \comp \inr] \\
    =&~ a \comp \mu \comp TT[\id, f] \comp T[\id, \eta \comp \inr] \\
    =&~ a \comp Ta \comp TT[\id, f] \comp T[\id, \eta \comp \inr] \\
    =&~ a \comp Ta \comp T[T[\id, f], \eta \comp f] \\
    =&~ a \comp T[a \comp T[\id, f], f] \\
    =&~ a \comp T[\id, f] \comp T[\inl \comp a \comp T[\id, f], \inr] \\
    =&~ \alpha \comp T(\alpha+ \id)
  \end{align*}
  \end{minipage}
\end{enumerate}
Next, we show that the passages (i) and (ii) are mutually inverse.
\begin{itemize}
\item From bialgebras to $\hash$-algebras and back: Given the bialgebra $TA
  \xrightarrow{~a~}A\xleftarrow{~f~}\Sigma A$, constructing $\alpha = a \comp T[\id,
  f]$ as in (ii), one obtains back $a$ and $f$ using (i):
  \begin{align*}
    \alpha \comp (T\inl) =&~ a \comp T[\id, f] \comp (T\inl) = a \\
    \alpha \comp \eta \comp \inr =&~ a \comp T[\id, f] \comp \eta \comp \inr = a \comp \eta \comp f = f
  \end{align*}
\item From $\hash$-algebras to bialgebras and back: Given $\alpha : T(A + \Sigma
  A) \to A$, we construct $a = \alpha \comp (T\inl)$ and $f = \alpha \comp \eta\inr$ and obtain:
  \begin{align*}
    a \comp T[\id, f] =&~ \alpha \comp (T\inl) \comp T[\id, \alpha \comp \eta \comp \inr] \\
    =&~ \alpha \comp (T\inl) \comp T[\alpha \comp u, \alpha \comp \eta \comp \inr] \\
    =&~ \alpha \comp (T\inl) \comp T[\alpha \comp \eta \comp \inl, \alpha \comp \eta \comp \inr] \\
    =&~ \alpha \comp (T\inl) \comp (T\alpha) \comp (T\eta) \\
    =&~ \alpha \comp T(\alpha + \id) \comp (T\inl) \comp (T\eta) \\
    =&~ \alpha \comp m \comp (T\inl) \comp (T\eta) \\
    =&~ \alpha \comp \mu \comp T[\id, \eta \comp \inr] \comp (T\inl) \comp (T\eta) \\
    =&~ \alpha \comp \mu \comp (T\eta) \\
    =&~ \alpha
  \end{align*}
\end{itemize}
\todo{extend to functors.}
\qed

\subsection*{Details for Example~\ref{ex:cpoalg}}

We verify the three axioms of complete Elgot $\IB{}{}$-algebras for $A$ equipped with the least solution. 
\begin{itemize}
\item solution: this clearly holds because $e^\istar$ is a fixed point of the map $s \mapsto a \comp (\IB{\id}{s}) \comp e$. 
\item uniformity: let $e: X \to \IB{A}{X}$, $f: Y \to \IB{A}{Y}$ and $h: X \to Y$ such that $fh = (\IB{\id}{h})e$ holds. In order to show that $f^\istar h = e^\istar$ we show by induction that for every $i$ we have 
  \[
    f_i^\istar h = e_i^\istar.
  \]
  The base case is clear: since composition is left strict we have $\bot h = \bot$. For the induction step we compute:
  \begin{align*}
    f_{i+1}^\istar h &= a \comp (\IB{\id}{f_i^\istar})\comp f \comp h \\
    &= a \comp (\IB{\id}{f_i^\istar})\comp(\IB{\id}{h})\comp e \\
    &= a \comp (\IB{\id}{f_i^\istar \comp h})\comp e \\
    &= a \comp (\IB{\id}{e_i^\istar})\comp e \\
    &= e_{i+1}^\istar.
  \end{align*}
  \def\fq{\filledsquare}
\item compositionality: let $f: Y \to \IB{A}{Y}$ and
  $g: X \to \IB{Y}{X}$. The desired equation
  $(f \fq g)^\istar\inr = (f^\istar \bullet g)^\istar$ is
  established by proving by induction the following two inequalities
  for every $i$:
  \begin{align}
    (f \fq g)_i^\istar\inr &\sqsubseteq (f^\istar \bullet g)^\istar\label{eq:ineq1} \\
    (f^\istar \bullet g)_i^\istar &\sqsubseteq (f \fq g)^\istar\inr\label{eq:ineq2}
  \end{align}
  For~\eqref{eq:ineq1}, the base case is clear by left strictness. For the induction step first note that $(f \fq g)\inl = (\IB{\id}{\inl})f$, thus $(f \fq g)^\istar\inl = f^\istar$ by uniformity. Now we consider the following diagram 
  \[\scriptsize
    \begin{tikzcd}
      X 
      \ar[rrrr, "(f^\istar \bullet g)^\istar"]
      \ar[rrd, "\inr"]
      \ar[dddd, "g"]
      & &\mbox{ } & &
      A
      \\
      &
      &
      Y+X
      \arrow{rru}[below]{(f\fq g)_{i+1}^\istar}
      \descto{u}{\begin{turn}{90}$\sqsubseteq$\end{turn}}
      \ar[d, "f \fq g"]
      \\
      &
      &
      \IB{A}{(Y+X)}
      \ar[r, "\IB{\id}{(f \fq g)_i^\istar}"]
      &
      \IB{A}{A}
      \ar[ruu, "a"]
      \\
      &
      \IB{(\IB{A}{(Y+X)})}{(Y+X)}
      \ar[ru, "m"]
      \ar[rr, "\IB{(\IB{\id}{(f \fq g)_i^\istar})}{(f \fq g)_i^\istar}"]
      &\mbox{ }&
      \IB{(\IB{A}{A})}{A}
      \ar[u, "m"]
      \ar[rdd, "\IB{a}{\id}"]
      \\
      \IB{Y}{X}
      \ar[r, "\IB{f}{\id}"]
      \ar[d, "\IB{f^\istar}{\id}"]
      &
      \IB{(\IB{A}{Y})}{X}
      \ar[u, "\IB{(\IB{\id}{\inl})}{\inr}"]
      \arrow{rru}[below,xshift=15pt]{\IB{(\IB{\id}{f^\istar})}{(f^\istar \bullet g)^\istar}}
      \descto{ru}{\begin{turn}{-77}$\sqsubseteq$\end{turn}}
      \\
      \IB{A}{X}
      \ar[rrrr, "\IB{\id}{(f^\istar \bullet g)^\istar}"]
      &&&&
      \IB{A}{A}
      \ar[uuuuu, "a"]
    \end{tikzcd}\normalsize
  \]
  We are to prove the inequality in the upper triangle. We start with
  the inequality in the middle triangle; it holds by the induction
  hypothesis and since
  $(f \fq g)_i^\istar \inl \sqsubseteq (f\fq g)^\istar \inl =
  f^\istar$. The other inner parts clearly commute; for the lower
  square consider the left- and right-hand components of $\IB{}{}$
  separately: the right-hand one commutes trivially, and for the
  left-hand one use the solution axiom for $f$. Since the outside of the 
  diagram also commutes by the solution axiom, we obtain the desired inequality 
  in the upper triangle. 

  For~\eqref{eq:ineq2}, the base case is clear once again. For the induction step we consider the diagram below:
  \[\scriptsize
    \begin{tikzcd}
      X 
      \ar[rrrr, "(f^\istar \bullet g)_{i+1}^\istar"]
      \ar[rrd, "\inr"]
      \ar[dddd, "g"]
      & &\mbox{ } & &
      A
      \\
      &
      &
      Y+X
      \arrow{rru}[below]{(f\fq g)^\istar}
      \descto{u}{\begin{turn}{-90}$\sqsubseteq$\end{turn}}
      \ar[d, "f \fq g"]
      \\
      &
      &
      \IB{A}{(Y+X)}
      \ar[r, "\IB{\id}{(f \fq g)^\istar}"]
      &
      \IB{A}{A}
      \ar[ruu, "a"]
      \\
      &
      \IB{(\IB{A}{(Y+X)})}{(Y+X)}
      \ar[ru, "m"]
      \ar[rr, "\IB{(\IB{\id}{(f \fq g)^\istar})}{(f \fq g)^\istar}"]
      &\mbox{ }&
      \IB{(\IB{A}{A})}{A}
      \ar[u, "m"]
      \ar[rdd, "\IB{a}{\id}"]
      \\
      \IB{Y}{X}
      \ar[r, "\IB{f}{\id}"]
      \ar[d, "\IB{f^\istar}{\id}"]
      &
      \IB{(\IB{A}{Y})}{X}
      \ar[u, "\IB{(\IB{\id}{\inl})}{\inr}"]
      \arrow{rru}[below,xshift=15pt]{\IB{(\IB{\id}{f^\istar})}{(f^\istar \bullet g)_i^\istar}}
      \descto{ru}{\begin{turn}{103}$\sqsubseteq$\end{turn}}
      \\
      \IB{A}{X}
      \ar[rrrr, "\IB{\id}{(f^\istar \bullet g)_i^\istar}"]
      &&&&
      \IB{A}{A}
      \ar[uuuuu, "a"]
    \end{tikzcd}
    \normalsize
  \]
  We are to prove the inequality in the upper triangle. We start with the inequality in the middle triangle; it holds by the induction hypothesis and since $(f \fq g)^\istar \inl = f^\istar$. The other inner parts commute as in the previous diagram. Finally, the outside commutes by the definition of $(f^\istar\bullet g)_{i+1}^\istar$. Thus, we obtain the inequality in the upper triangle as desired. \qed
\end{itemize}
\subsection*{Proof of Proposition~\ref{prop:hom}}
    Let $e = (\IB{\id}{\inr}) \comp [\id, u_A^A] : (\IB{A}{A}) + A \to
    \IB{A}{((\IB{A}{A})+A)}$. We show that $e^{\istar}=[a, \id]$, and therefore $e^{\istar} \comp \inl = a$. To that end we successively  calculate $e^{\istar} \comp \inr$ and $e^{\istar} \comp \inl$:
    \begin{flalign*}
      &&e^{\istar} \comp \inr =&~ a \comp (\IB{\id}{e^{\istar}}) \comp e \comp \inr &\by{solution} \\
      &&=&~ a \comp (\IB{\id}{e^{\istar}}) \comp (\IB{\id}{\inr}) \comp u_A^A \\
      &&=&~ a \comp (\IB{\id}{e^{\istar}}) \comp u_A^{(\IB{A}{A}) + A} \comp \id \\
      &&=&~ a \comp u_A^A \comp \id \\
      &&=&~ \id.\\
      &&e^{\istar} \comp \inl =&~ a \comp (\IB{\id}{e^{\istar}}) \comp e \comp \inl &\by{solution} \\
      &&=&~ a \comp (\IB{\id}{e^{\istar}}) \comp (\IB{\id}{\inr}) \\
      &&=&~ a \comp (\IB{\id}{(e^{\istar} \comp \inr)}) \\
      &&=&~ a \comp (\IB{\id}{\id}) \\
      &&=&~ a.
    \end{flalign*}
    To finish the proof, let us show that the following diagram:
    \[
      \begin{tikzcd}[column sep=huge, row sep=huge]
        \IB{A}{A}
          \ar[r, "\inl"]
          \ar[d, "\IB{f}{f}", swap] &
        (\IB{A}{A}+A)
          \ar[dr, sloped, pos=0.7, "{((\IB{f}{\id})\comp e)^{\iistar}}"']
          \ar[r, "{e^{\istar}=[a,\id]}"] &
        A
          \ar[d, "f"] \\
        \IB{B}{B}
          \ar[rr, "b"] &&
        B
      \end{tikzcd}
    \]
    commutes.
    The triangle commutes since $f$ is a morphism of $\IB{}{}$-algebras, and we are left to show commutativity of the inner quadrangle. Observe that
    \begin{flalign*}
      &&((\IB{f}{\id}) \comp e)^{\iistar} \comp \inr =&~ b \comp
      (\IB{\id}((\IB{f}{\id})\comp e)^{\iistar}) \comp (\IB{f}{\id}) \comp e
      \comp \inr & \by{solution} \\
      && =&~ b \comp (\IB{\id}{((\IB{f}{\id})\comp e)^{\iistar}}) \comp
      (\IB{f}{\inr}) \comp u_A^A \\
      && =&~ b \comp (\IB{f}{((\IB{f}{\id})\comp e)^{\iistar}}) \comp
      u_A^{(\IB{A}{A})+A} \\
      && =&~ b \comp u_B^B \comp f \\
      && =&~ f
\intertext{from which we conclude the desired identity:}
      &&((\IB{f}{\id}) \comp e)^{\iistar} \comp \inl =&~ b \comp
      (\IB{\id}((\IB{f}{\id})\comp e)^{\iistar}) \comp (\IB{f}{\id}) \comp e
      \comp \inl & \by{solution} \\
      && =&~ b \comp (\IB{\id}{((\IB{f}{\id})\comp e)^{\iistar}}) \comp
      (\IB{f}{\inr}) \\
      && =&~ b \comp (\IB{f}{(((\IB{f}{\id})\comp e)^{\iistar} \comp \inr)}) \\
      && =&~ b \comp (\IB{f}{f}).&\text{\qed}
    \end{flalign*}

\subsection*{Proofsketch for Theorem~\ref{thm:eq}}

Before we outline the proof of the desired result, we explain an auxiliary construction that produces from a given complete Elgot algebra $A$ and morphism $f: Y \to A$ a new complete Elgot algebra on $\IB Y A$. 
\begin{constr}\label{c:alg}
  Let $(A, \alpha, \argument^\dagger)$ a complete Elgot $\IB{}{}$-algebra and let $m: Y \to A$ be a morphism. Then form the following morphism
  \[
    \alpha^f = \Bigl(\IB{(\IB Y A)}{(\IB Y A)} \xrightarrow{\IB{\id}{(\IB {f}{\id})}} \IB{(\IB Y A)}{(\IB A A)} \xrightarrow{\IB{\id}{\alpha}}
    \IB{(\IB Y A)} A \xrightarrow{m^A_Y} \IB Y A\Bigr)
  \]
  and define the dagger operation $\argument^\ddagger$ as follows: given $e: X \to \IB{(\IB Y A)} X$ one forms 
  \[
    \ol e = (X \xrightarrow{e} \IB{(\IB Y A)}X \xrightarrow{\IB{(\IB{f} {\id})} {\id}} \IB{(\IB A A)}X \xrightarrow{\IB {\alpha}{\id}} \IB A X), 
  \]
  i.e.~$\ol e = (\alpha (\IB{f}{\id})) \bullet e$, and then one puts
  \[
    e^\ddagger = \Bigl(X \xrightarrow{e} \IB{(\IB Y A)}X \xrightarrow{\IB{\id}{\ol e}^\dagger} \IB{(\IB Y A)}A \xrightarrow{m^A_Y} \IB Y A\Bigr).
  \]
\end{constr}
\begin{lem}\label{lem:aux}
  The triple $(\IB Y A, \alpha^f, \argument^\ddagger)$ is a complete Elgot $\IB{}{}$-algebra such that
  \[
    \IB Y A \xrightarrow{\IB f {\id}} \IB A A \xrightarrow{\alpha} A
  \]
  is a morphism of complete Elgot algebras.
\end{lem}
The proof is a somewhat involved computation adapted from the proof of~\cite[Lemma~5.6]{AdamekMiliusEtAl06j}.\smnote{We will type it later.}
The previous lemma and the following proposition provide a Lambek-type lemma for complete Elgot $\IB{}{}$-algebras. 
\begin{prop}\label{prop:iso}
  If $(FY, \varphi_Y, \argument^\dagger)$ is a free complete Elgot $\IB{}{}$-algebra on $Y$ with universal morphism $\eta_Y: Y \to FY$, then 
  \[
    \IB Y {FY} \xrightarrow{\IB{\eta_Y}{\id}} \IB{FY}{FY} \xrightarrow{\varphi_Y} FY
  \]
  is an isomorphism.
\end{prop} 
\begin{proof}
  By Lemma~\ref{lem:aux}, $\IB Y {FY}$ with algebra structure and
  $\argument^\ddagger$ formed as in Construction~\ref{c:alg} (for $A = FY$,
  $\alpha = \varphi_Y$ and $f = \eta_Y$) is a complete Elgot $\IB{}{}$-algebra.
  By the freeness of $FY$, we obtain a unique morphism of complete Elgot
  $\IB{}{}$-algebras $t: FY \to \IB Y {FY}$ such that $t \eta_Y = u^{FY}_Y: Y
  \to \IB Y{FY}$. Let us denote $t' = \varphi_Y(\IB{\eta_Y}{\id})$. Then it is
  our task to prove that $t$ and $t'$ are mutually inverse.

  Indeed, we have $t' \comp t = \id_{FY}$ since both $t$ and $t'$ are morphisms
  of complete Elgot $\IB{}{}$-algebras and since
  \[
    t' \comp t \comp \eta_Y = t'\comp u^{FY}_Y = \varphi_Y \comp
    (\IB{\eta_Y}{\id}) \comp u^{FY}_Y = \varphi_Y \comp u^{FY}_{FY} \comp \eta_Y = \eta_Y
  \]
  using naturality of $u$ and the unit law of the $\IB{}{}$-algebra structure
  $\varphi_Y$. The freeness of $FY$ now yields the desired equation.

  In order to prove $t \comp t' = \id_{\IB Y {FY}}$ notice first that $t$ is a
  morphism of $\IB{}{}$-algebras by Proposition~\ref{prop:hom}. Now consider the
  diagram below:
  \[
    \begin{tikzcd}
      \IB Y {FY}
      \ar[d, "\IB{\eta_Y}{\id}"] \ar[dddd, shiftarr={xshift=-30pt}, "t'", swap]
      \ar[rr, "\IB {\id} t"]
      &&
      \IB Y {(\IB Y {FY})}
      \ar[ld, "\IB{u^{FY}_Y}{\id}",swap]
      \ar[dddd, shiftarr={xshift=45pt}, "\IB {\id}{t'}"]
      \ar[dd, "\IB {\id}{(\IB{\eta_Y}{\id})}",swap]
      \\
      \IB{FY}{FY}
      \ar[r, "\IB t t"]
      \ar[ddd, "\varphi_Y"]
      &
      \IB{(\IB Y {FY})}{(\IB Y {FY})}
      \ar[d, "\IB{\id}{(\IB{\eta_Y}{\id})}"]
      \\
      &
      \IB{(\IB Y {FY})}{(\IB{FY}{FY})}
      \ar[d,"\IB{\id}{\varphi_Y}"]
      &
      \IB Y {(\IB{FY}{FY})}
      \ar[l, "\IB{u^{FY}_{FY}}{\id}"]
      \ar[dd, "\IB {\id}{\varphi_Y}"]
      \\
      &
      \IB{(\IB Y {FY})}{FY}
      \ar[d, "m^{FY}_Y"]
      \\
      FY
      \ar[r, "t"]
      \ar[rr, shiftarr={yshift=-15pt}, "t", swap]
      &
      \IB Y {FY}
      &
      \IB Y {FY}
      \ar[lu, "\IB{u^{FY}_Y}{\id}",swap]
      \ar[l, "\id"]
    \end{tikzcd}
  \]
  All its inner parts commute: the left- and right-hand parts commute by the
  definition of $t'$, the big lower left-hand part commutes since $t$ is a
  $\IB{}{}$-algebra morphism, the upper part commutes using that $t \comp \eta_Y
  = u^{FY}_Y$, the lower right-hand triangle commutes by the monad laws for
  $\IB{\argument}{FY}$ and the remaining three parts are obvious. Thus, the
  outside commutes, and since $t' \comp t = \id$ we conclude that $t \comp t' =
  \id$, which completes the proof. 
\end{proof}
Let us henceforth denote for a given free complete Elgot $\IB{}{}$-algebra on $Y$
\[
  t = (\varphi_Y (\IB{\eta_Y}{\id}))^{-1}: FY \to \IB Y {FY}.
\]
\begin{cor}\label{cor:t}
  The following square commutes:
  \[
    \begin{tikzcd}
      \IB{FY}{FY} \ar[r, "\varphi_Y"]\ar[d, "\IB{t}{\id}"] & FY \ar[d, "t"] \\
      \IB{(\IB Y {FY})}{FY} \ar[r, "m^{FY}_Y"] & \IB Y {FY}
    \end{tikzcd}
  \]
\end{cor}
\begin{proof}
  Use that $t$ is a morphism of $\IB{}{}$-algebras, i.e.\ consider the big lower
  left-hand part of the diagram in the proof of Proposition~\ref{prop:iso}. Then
  use that $\varphi_Y \comp (\IB{\eta_Y}{\id}) \comp t = \id$ implies that
  \[
    (\IB{\id}{\varphi_Y})\comp (\IB{\id}{(\IB{\eta_Y}{\id})}) \comp (\IB t t) = \IB{t}{\id}
  \]
  to obtain the commutativity of the desired square. 
\end{proof}
We are now ready to prove item 2.~of Theorem~\ref{thm:eq}. So suppose that $(FY,
\varphi_Y, \argument^\dagger)$ is a free complete Elgot $\IB{}{}$-algebra on
$Y$ with universal morphism $\eta_Y: Y \to FY$. By Lemma~\ref{lem:aux}, we know
that $\varphi_Y \comp (\IB{\eta_Y}{\id}): \IB Y {FY} \to FY$ is an isomorphism
with inverse $t: FY \to \IB Y{FY}$. One then proves that $(FY, t)$ is a final
coalgebra for $\IB{Y}{\argument}$.

Indeed, given any coalgebra $c: X \to \IB Y X$ one forms the equation morphism
\[
e = (X \xrightarrow{c} \IB Y X \xrightarrow{\IB{\eta_Y}{\id}} \IB{FY}X).
\] 
Then it is easy to see that $e^\dagger: X \to FY$ is a coalgebra homomorphism
from $(X,c)$ to $(FY, t)$; in fact, consider the diagram below:
\[
\begin{tikzcd}
  X \ar[rr, "e^\dagger"] \ar[dd,"c",swap] \ar[rd, "e"] && FY \ar[dd, shiftarr={xshift=30pt}, "t"]\\
  &
  \IB {FY} X \ar[r, "\IB{FY}{e^\dagger}"] & \IB{FY}{FY} \ar[u, "\varphi_Y"] \\
  \IB Y X \ar[ru,"\IB{\eta_Y}{\id}"] \ar[rr, "\IB {\id}{e^\dagger}", swap] && \IB Y {FY} \ar[u, "\IB{\eta_Y}{\id}"]
\end{tikzcd}
\]
All its inner parts commute: the upper part commutes since $e^\dagger$ is a
solution of $e$, the left-hand triangle commutes by the definition of $e$, the
lower part commutes trivially, and for the right-hand part use that $t$ is the
inverse of $\varphi_Y \comp (\IB{\eta_Y}{\id})$.
It remains to prove that uniqueness of
a coalgebra homomorphism from $(X,c)$ to $(FY,t)$. This proof can be performed
analogously to the proof of part (2) $\Rightarrow$ (1)
of~\cite[Theorem~5.4]{AdamekMiliusEtAl06j}. \qed

\subsection*{Proof of Proposition~\ref{prop:unique}}

  Recall first from Theorem~\ref{thm:eq} that $FY$ is (equivalently) a final $(\IB{Y}{\argument})$-coalgebra with the structure $t:FY \to \IB{Y}{FY}$ obtained as an inverse of
  \[
    \IB Y {FY} \xrightarrow{\IB{\eta_Y}{FY}} \IB{FY}{FY} \xrightarrow{\varphi_Y} FY.
  \]
  Let $e: X \to \IB {FY} X$ and consider the following $(\IB{Y}{\argument})$-coalgebra
  \[
    \begin{tikzcd}
      \ol e = (\IB{FY} X \ar[r, "\IB t e"] & 
      \IB{(\IB Y {FY})}{(\IB{FY}X)} \ar[d, "\IB{(\IB Y {u^X_{FY}})}{(\IB{FY}{Y})}"] \\ 
      & \IB{(\IB{Y}{(\IB{FY}X)})}{(\IB{FY} X)} \ar[r, "m^{\IB{FY}X}_Y"] & \IB{Y}{(\IB{FY}X)}).
    \end{tikzcd}
  \]
  Now let $d: X \to FY$ be any solution of $e$, i.e.~we have $d = \varphi_Y (\IB{FY} d) e$.
  We will prove below that $\varphi_Y(\IB{FY}{d}): \IB{FY} X \to FY$ is a coalgebra homomorphism from $\ol e$ to $t$. Since $\ol e$ does not depend on the solution $d$ we then conclude that 
  \[
    e^\dagger = \varphi_Y(\IB{FY}{e^\dagger})e = \varphi_Y(\IB{FY}{d})e = d
  \]
  using finality of $FY$ in the middle step. 
  
  To finish the proof consider the following diagram:
  \[
    \hspace*{-4cm}
    \begin{tikzcd}
      \IB{FY}X
      \ar[dd, "\IB t e"]
      \ar[r, "\IB{\id} d"]
      &
      \IB{FY}{FY}
      \ar[d, "\IB t {\id}"]
      \ar[rr, "\varphi_Y"]
      &&
      FY
      \ar[dddd, "t"]
      \\
      &
      \IB{(\IB Y {FY})}{FY}
      \ar[rrddd, bend left=25, "m^{FY}_Y"]
      \\
      \IB{(\IB{Y}{FY})}{(\IB{FY}{X})}
      \ar[r, "\IB{\id}{(\IB{\id}{d})}"]
      \ar[d, "\IB{(\IB{\id}{u^X_{FY}})}{\id}"]
      &
      \IB{(\IB{Y}{FY})}{(\IB{FY}{FY})}
      \ar[u, "\IB{\id}{\varphi_Y}"]
      \ar[rd, swap, "\IB{(\IB{\id}{u^{FY}_{FY}})}{\id}"]
      \\
      \IB{(\IB{Y}{(\IB{FY}{X})})}{(\IB{FY}{X})}
      \ar[d, "m^{\IB{FY}{X}}_Y"]
      \ar[rr, swap, "\IB{(\IB{\id}{(\IB{\id}{d})})}{(\IB{\id}{d})}"]
      &&
      \IB{(\IB{Y}{(\IB{FY}{FY})})}{(\IB{FY}{FY})}
      \ar[luu, swap, "\IB{(\IB Y{\varphi_Y})}{\varphi_Y}"]
      \ar[d, "m^{\IB{FY}{FY}}_Y"]
      \\
      \IB{Y}{(\IB{FY}{X})}
      \ar[rr, swap, "\IB {\id}{(\IB{\id}{d})}"]
      &&
      \IB{Y}{(\IB{FY}{FY})}
      \ar[r, swap, "\IB{\id}{\varphi_Y}"]
      &
      \IB{Y}{FY}
    \end{tikzcd}
  \]
  Note first that the left-hand edge is $\ol e$. The upper left-hand square commutes since $d$ is a solution of $e$, for the part below it use that $(\IB{FY} d) u^X_{FY} = u^{FY}_{FY}$ holds since $\IB{FY}d$ is a monad morphism, and the lower left-hand part commutes by the laws of $\IB{}{}$. The upper right-hand part commutes by Corollary~\ref{cor:t}, and the remaining little inner triangle commutes since $\varphi_Y u^{FY}_{FY} = \id_{FY}$ since $\varphi_Y$ is the structure of a $\IB{}{}$-algebra. Hence $\varphi_Y(\IB{FY}d)$ is a coalgebra homomorphisms as desired, which completes the proof.\qed

\subsection*{Proof of Lemma~\ref{lem:IBcorec}}
To show the claim, form the following coalgebra for $\IB{A}{\argument}$:
\[
  \begin{tikzcd}
    X \ar[r, "e"] & \IB B X \ar[r, "\IB f X"] & \IB{(\IB{A}{\IBnu A})}{X}
    \ar[rr, "\IB{(\IB{A}{\inl})}{\inr}"] && \IB{(\IB{A}{(\IBnu A + X}))}{(\IBnu A + X)}
    \ar[d, "m^{\IBnu A + X}_A"] \\
    &&&& \IB{A}{(\IBnu A + X)}
  \end{tikzcd}
\]
By Proposition~\ref{prop:corec} we obtain a unique $h: X \to \IBnu A$ such that the diagram below commutes:
\[
  \begin{tikzcd}[column sep=10.8em]
    X
    \arrow[d, "h"']
    \arrow[r, "{m_A^{\IBnu{A} + X} \comp (\IB{((\IB{\id}{\inl}) \comp
      f)}{\inr}) \comp e}"]
    & \IB{A}{(\IBnu{A} + X)}
    \arrow[d, "{\IB{\id}{[\id, h]}}"] \\
    \IBnu{A}
    \arrow[r, "\out"]
    & \IB{A}{\IBnu{A}}.
  \end{tikzcd}
\]
Now use that $\IB{\id}{[\id, h]}$ is a monad morphism to see that, equivalently, $h$ is unique such that~\eqref{eq:corec} commutes. 
\qed

\subsection*{Proof of Lemma~\ref{lem:corec-coit}}
Let $g = (\tuo \comp f)^{\klstar} = \corec((\IB{(\tuo \comp f)}{\id}) \comp
\out, \out)$. Then we have
  \[
    \begin{tikzcd}[column sep=8.6em]
      X
      \arrow[r, "e"]
      \arrow[d, "\coit(e)"]
      & \IB{B}{X}
      \arrow[d, "\IB{\id}{\coit(e)}"] \\
      \IBnu{B}
      \arrow[r, "\out"]
      \arrow[d, "g"]
      & \IB{B}{\IBnu{B}}
      \arrow[r, "\IB{(\tuo \comp f)}{\id}"]
      & \IB{\IBnu{A}}{\IBnu{A}}
      \arrow[d, "m_A^{\IBnu A} \comp (\IB{\out}{g})"] \\
      \IBnu{A}
      \arrow[rr, "\out"]
      & & \IB{A}{\IBnu{A}}
    \end{tikzcd}
  \]
  and the uppermost path from $X$ to $\IB{A}{\IBnu{A}}$ amounts to 
  $m_A^{\IBnu A} \comp (\IB{f}{(g \comp \coit(e))})$. Therefore, $g \comp \coit(e)$ satisfies the equation uniquely
  determining $\corec(e, f)$, implying the result. 
\qed
\subsection*{Proof of Lemma~\ref{lem:corec-nat}}
  Notice that diagram \eqref{eq:corec} implies trivially that
  \[ \corec(e, f) = \corec((\IB{f}{\id})\comp e, \id). \]
 Thus we get 
  \begin{flalign*}
    && &~ (\IBnu g) \comp \coit(e) \\
    && &~ (\etan \comp g)^{\klstar} \comp \coit(e) \\
    &&=&~ (\tuo \comp \out \comp \etan \comp g)^{\klstar} \comp \coit(e) \\
    &&=&~ \coit(e, \out \comp \etan \comp g) &\by{Lemma~\ref{lem:corec-coit}} \\
    &&=&~ \coit(e, u \comp g) \\
    &&=&~ \coit((\IB{g}{\id}) \comp e, u) & \by{definition of $\corec(\argument,\argument)$} \\
    &&=&~ \coit((\IB{g}{\id}) \comp e).&\by{corollary to Lemma~\ref{lem:corec-coit}}\quad \text{\qed}
  \end{flalign*}

\noindent
Let introduce the following useful morphism:
\begin{align*}
\ext=\tuo\comp (\IB{\id}{\etan}): \IB{X}{X}\to \IBnu X
\end{align*}
natural in $X$.
\subsection*{Proof of Theorem~\ref{thm:algebras-iso}}
  The proof is organized as follows. First we construct for each
  $\IBnu$-algebra a complete Elgot $\IB{}{}$-algebra  and
  vice versa. Then we extend these constructions to functors and prove that these
  functors witness an isomorphism of categories.

  Given an algebra $\chi : \IBnu{A} \to A$, i.e.
  \[
    \begin{tikzcd}
      A
      \arrow[rd, "\id"']
      \arrow[r, "\etan"]
      & \IBnu{A}
      \arrow[d, "\chi"] \\
      & A
    \end{tikzcd}\hspace{3cm}\begin{tikzcd}
      \IBnu{\IBnu{A}}
      \arrow[d, "\mun"']
      \arrow[r, "\IBnu{\chi}"]
      & \IBnu{A}
      \arrow[d, "\chi"] \\
      \IBnu{A}
      \arrow[r, "\chi"]
      & A
    \end{tikzcd}
  \]
  we define a $\IB{}{}$-algebra $a  : \IB{A}{A} \to A$ as follows:
  \[
    \IB{A}{A} \xrightarrow{~~\ext~~} \IBnu{A} \xrightarrow{~~\chi~~} A
  \]
  It is easy to see that $a $ is an algebra for the monad
  $\IB{\argument}{A}$. For any $e : X \to \IB{A}{X}$, let $e^{\istar}:X\to A$ be
  given by $\chi \comp (\coit e)$.
  
  We now need to check if the so-defined iteration operator satisfies the
  axioms of complete Elgot algebras.

  \paragraph{Solution.}
  To see that this holds, consider the following diagram, the outside of which
  constitutes the required property:
  \[
    \begin{tikzcd}[row sep=large, column sep=12em]
      X
      \arrow[d, "\coit e"]
      \arrow[dd, shiftarr={xshift=-30pt}, "e^{\istar}"']
      \arrow[r, "e"]
      & \IB{A}{X}
      \arrow[d, "\IB{\id}{\coit e}"']
      \arrow[dd, shiftarr={xshift=30pt}, "\IB{\id}{e^{\istar}}"] \\
      \IBnu{A}
      \arrow[r, shift left, "\out"]
      \arrow[d, "\chi"]
      & \IB{A}{\IBnu{A}}
      \arrow[d, "\IB{\id}{\chi}"']
      \arrow[l, shift left, "{\tuo}"] \\
      A
      & \IB{A}{A}
      \arrow[l, "a "]
    \end{tikzcd}
  \]
  The top square obviously commutes as a finality diagram. For the lower square,
  we calculate
  \begin{align*}
    &~\chi \comp \tuo \\
    =&~\chi \comp \mun \comp \etan\comp \tuo \\
    =&~\chi \comp \mun \comp \tuo \comp (\IB{\etan}{\etan}) \\
    =&~\chi \comp \IBnu{\chi} \comp \tuo \comp (\IB{\etan}{\etan}) \\
    =&~\chi \comp \tuo \comp (\IB{\chi}{\IBnu{\chi}}) \comp (\IB{\etan}{\etan}) \\
    =&~\chi \comp \tuo \comp (\IB{\id}{\etan \comp \chi}) \\
    =&~\chi \comp \tuo \comp (\IB{\id}{\etan}) \comp (\IB{\id}{\chi}) \\
    =&~a  \comp (\IB{\id}{\chi}).
  \end{align*}

  \paragraph{Functoriality.}
  This is a simple consequence of the definition of the dagger operation in terms
  of $\coit$. Suppose that $f\comp h = (\IB{\id}{h}) \comp e$. Then
  \begin{align*}
    &~\out \comp (\coit f) \comp h \\
    =&~(\IB{\id}{(\coit f)}) \comp f \comp h \\
    =&~(\IB{\id}{(\coit f)}) \comp (\IB{\id}{ h}) \comp e \\
    =&~(\IB{\id}{((\coit f) \comp  h)}) \comp e,
  \end{align*}
  i.e.\ $(\coit f) \comp h$ satisfies the identity uniquely characterizing $\coit e$.
  Therefore $(\coit f) \comp h=\coit(e)$ and hence $f^{\istar} \comp h =\chi\comp (\coit f) \comp h = \chi\comp(\coit e)=e^{\istar}$.

  \paragraph{Compositionality.}
  We have on the one hand
  \begin{flalign*}
    &&(f^{\istar} \bullet g)^{\istar} =&~ \chi \comp \coit(f^{\istar} \bullet g) \\
    &&=&~\chi \comp \coit\left((\IB{\chi \comp (\coit f)}{\id}) \comp g\right) \\
    &&=&~\chi \comp \coit\left((\IB{\chi}{\id}) \comp (\IB{(\coit f)}{\id}) \comp g\right) \\
    &&=&~\chi \comp (\IBnu{\chi}) \comp \coit((\IB{(\coit f)}{\id}) \comp g) & \by{Lemma~\ref{lem:corec-nat}} \\
    &&=&~\chi \comp \mun \comp \coit((\IB{(\coit f)}{\id}) \comp g)  & \by{$\chi$ is an $\IBnu$-algebra} \\
    &&=&~\chi \comp \corec((\IB{(\coit f)}{\id}) \comp g, \out), & \by{Lemma~\ref{lem:corec-coit}}
  \end{flalign*}
and on the other hand,
\begin{flalign*}
    && (f \filledsquare g)^{\istar} \comp \inr =&~ \chi \comp \coit(f \filledsquare g) \comp \inr \\
    && =&~\chi \comp \coit(m^{Y+X}_A \comp (\IB{((\IB{\id}{\inl}) \comp f)}{\inr}) \comp [u_Y^X, g]) \comp \inr.&
  \end{flalign*}
  Now, let $h = m^{Y+X}_A \comp (\IB{((\IB{\id}{\inl}) \comp f)}{\inr}) \comp
  [u_Y^X, g]$. We are finished once we proved that $\coit(h) \comp \inr$
  satisfies the identity characterizing  $\corec((\IB{\coit(f)}{\id}) \comp
  g, \out)$. First observe that the following:
\begin{flalign*}
    &&&\out \comp (\coit h) \comp \inl& \\
    &&=&~(\IB{\id}{\coit h}) \comp h \comp \inl \\
    &&=&~(\IB{\id}{\coit h}) \comp m^{Y+X}_A \comp (\IB{((\IB{\id}{\inl}) \comp f)}{\inr}) \comp u_Y^X\\
    &&=&~(\IB{\id}{\coit h}) \comp m^{Y+X}_A \comp u^{Y+X}_A \comp (\IB{\id}{\inl}) \comp f\\
    &&=&~(\IB{\id}{\coit h}) \comp (\IB{\id}{\inl}) \comp f \\
    &&=&~(\IB{\id}{((\coit h) \comp \inl)}) \comp f,
\end{flalign*}
i.e.\ $(\coit h) \comp \inl$ satisfies the identity chracterizing $\coit(f)$ and therefore
\begin{align}\label{eq:h_inl_f}
(\coit h) \comp \inl = \coit f 
\end{align}
Then we proceed as follows:
  \begin{flalign*}
    &&&~\out \comp (\coit h) \comp \inr \\
    &&=&~\out \comp \coit(m^{Y+X}_A \comp (\IB{((\IB{\id}{\inl}) \comp f)}{\inr}) \comp
       [u_Y^X, g]) \comp \inr \\
    &&=&~(\IB{\id}{\coit h}) \comp m^{Y+X}_A (\IB{((\IB{\id}{\inl}) \comp f)}{\inr}) \comp g \\
    &&=&~m^{\IBnu{A}}_A \comp (\IB{((\IB{\id}{((\coit h) \comp \inl)}) \comp f)}{((\coit h) \comp \inr)}) \comp g.&\\
    &&=&~m^{\IBnu{A}}_A \comp (\IB{((\IB{\id}{\coit f}) \comp f)}{((\coit h) \comp \inr)}) \comp g&\by{\eqref{eq:h_inl_f}} \\
    &&=&~m^{\IBnu{A}}_A \comp (\IB{(\out \comp (\coit f))}{((\coit h) \comp \inr))} \comp g \\
    &&=&~m^{\IBnu{A}}_A \comp (\IB{\out}{((\coit h) \comp \inr)}) \comp (\IB{\coit f}{\id}) \comp g
  \end{flalign*}
  which, as indicated above, implies that $(\coit h) \comp \inr = \corec((\IB{\coit f}{\id}) \comp g, \out)$.
  
  We proceed with the converse construction: Given a complete Elgot $\IB{}{}$-algebra $a  : \IB{A}{A} \to A$, we build 
a $\IBnu{}$-algebra by iterating the structure of the final coalgebra $\out : \IBnu{A} \to \IB{\IBnu{A}}{A}$:
\begin{align*}
\IBnu{A} \xrightarrow{~~\out^{\istar}~~}  A
\end{align*}
  To show that $\out^{\istar}$ is an $\IBnu{}$-algebra, we check the following.

  \paragraph{Compatibility with unit.} Since
  $
    \out \comp \etan = u_A^{\IBnu{A}} = (\IB{A}{\etan}) \comp u_A^A
  $, by the functoriality axiom,
  \[
    \out^{\istar} \comp \etan = (u_A^A)^{\istar}.
  \]
  Using the solution axiom, we obtain, since $a$ is an $(\IB{\argument}{A})$-algebra,
  \[
    \out^{\istar} \comp \etan =(u_A^A)^{\istar} = a  \comp (\IB{\id}{(u_A^A)^{\istar}}) \comp u_A^A =
    a  \comp u_A^A = \id.
  \]

   \paragraph{Compatibility with multiplication.}
   \mbox{We need to show that $\out^{\istar} \comp \IBnu{(\out^{\istar})} = \out^{\istar}
   \mun$}.
   
   Note that the type of morphisms on the left and on the right hand sides 
   is $\IBnu{\IBnu{A}} \to A$. We show that both morphisms are equal to 
   $(\out^{\istar} \bullet \out)^{\istar}$ having the same type, 
   which is itself by compositionality equal to $(\out \filledsquare
   \out)^{\istar} \comp \inr$. 
 
   For the left-hand side of the original equation we obtain this by functoriality:
   \begin{align*}
     \out \comp \IBnu{(\out^{\istar})}
     =&~ (\IB{\id}{(\IBnu{\out^{\istar}})}) \comp (\IB{\out^{\istar}}{\id}) \comp \out \\
     =&~ (\IB{\id}{(\IBnu{\out^{\istar}})}) \comp (\out^{\istar} \bullet \out)
   \end{align*}
   and therefore
   \[
     \out^{\istar} \comp \IBnu{(\out^{\istar})} = (\out^{\istar} \bullet \out)^{\istar}.
   \]
   As for the right-hand side, consider the following diagram:
\begin{equation}\label{eq:c_mult}
     \begin{tikzcd}[column sep=huge]
       \IBnu{A} + \IBnu{\IBnu{A}}
       \arrow[r, "\out \filledsquare \out"]
       \arrow[d, "{[\id, \mun]}"']
       & \IB{A}{(\IBnu{A} + \IBnu{\IBnu{A}})}
       \arrow[d, "{\IB{\id}{[\id, \mun]}}"] \\
       \IBnu{A}
       \arrow[r, "\out"']
       & \IB{A}{\IBnu{A}}
     \end{tikzcd}
\end{equation}   
Let us verify that this diagram commutes by case distinction (we drop the indices at $m$ and $u$
   for readability).   
On the one hand,
     $\out \comp [\id, \mun] \comp \inl = \out$
   and also
   \begin{flalign*}
     &&&~(\IB{\id}{[\id,\mun]}) \comp (\out \filledsquare \out) \comp \inl \\
     &&=&~(\IB{\id}{[\id,\mun]}) \comp m \comp (\IB{((\IB{\id}{\inl}) \comp \out)}{\inr}) \comp u&\by{definition of~$\filledsquare$} \\
     &&=&~(\IB{\id}{[\id,\mun]}) \comp m \comp u \comp (\IB{\id}{\inl}) \comp \out \\
     &&=&~(\IB{\id}{[\id,\mun]}) \comp (\IB{\id}{\inl}) \comp \out \\
     &&=&~\out;
   \end{flalign*}
   analogously, on the other hand,
   \begin{flalign*}
     &&\out \comp [\id, &\mun] \comp \inr\\
     &&=&~\out \comp \mun&\\
     &&=&~m \comp (\IB{\out}{\mun}) \comp \out
   \intertext{and}
     &&(\IB{\id}{[\id,&\mun]}) \comp (\out \filledsquare \out) \comp \inr \\
     &&=&~(\IB{\id}{[\id,\mun]}) \comp m \comp (\IB{((\IB{\id}{\inl}) \comp \out)}{\inr}) \comp \out \\
     &&=&~m \comp (\IB{(\IB{\id}{[\id,\mun]})}{[\id,\mun]}) \comp (\IB{((\IB{\id}{\inl}) \comp \out)}{\inr}) \comp \out \\
     &&=&~m \comp (\IB{(\IB{\id}{[\id,\mun] \comp \inl}) \comp \out}{[\id,\mun]\comp \inr}) \comp \out \\
     &&=&~m \comp (\IB{\out}{\mun}) \comp \out.
   \end{flalign*}
From~\eqref{eq:c_mult}, by functoriality, we obtan
   \[ 
     \out^{\istar} \comp [\id, \mun] = (\out \filledsquare \out)^{\istar}
   \]
   and thus
   \[
     \out^{\istar} \comp \mun = (\out \filledsquare \out)^{\istar} \comp \inr.
   \]
   
   Let us now complete the constructed correspondence between $|\BC^{\IBnu}|$
   and $|\cbElg(\BC)|$ to an equivalence of categories. Let 
   $F : \BC^{\IBnu} \to \cbElg(\BC)$ be defined as follows: $F$ assigns
   to an $\IBnu$-algebra $(A, \chi)$ the complete Elgot algebra $(A,
   \chi\comp\ext, \argument^\istar)$ with the iteration as presented above, and
   to an $\IBnu$-algebra homomorphism $f : (A, \chi) \to (B, \zeta)$ the underlying morphism
from $A$ to $B$. Let us check that this definition is correct, i.e.\ the above $f$ is a complete 
   Elgot $\IB{}{}$\dash algebra morphism from $F(A, \chi) = (A, a, \argument^\istar)$ to
   $F(B, \zeta) = (B, b, \argument^\iistar)$, i.e.\ for any $e : X \to
   \IB{A}{X}$:
   \begin{flalign*}
     &&f \comp e^{\istar} =&~ f \comp\chi \comp (\coit e) & \by{definition of $\argument^\istar$} \\
     &&=&~ \zeta \comp (\IBnu{f}) \comp (\coit e) & \by{$f$ is a $\IBnu$-algebra morphism} \\
     &&=&~ \zeta \comp \coit\left(\left(\IB{f}{\id}\right) \comp e\right) & \by{Lemma~\ref{lem:corec-nat}} \\
     &&=&~ \left(\left(\IB{f}{\id}\right) \comp e\right)^{\iistar} \\
     &&=&~ (f \bullet e)^{\iistar}.
   \end{flalign*}
   For the converse direction, let $G : \cbElg(\BC) \to \BC^{\IBnu}$
   send a complete Elgot $\IB{}{}$-algebra $(A, a, \istar)$ to $(A,
   \out^{\istar})$, which we proved to be an $\IBnu$-algebra. 
   Given a morphism $f : (A, a, \istar) \to (B, b, \iistar)$, let $G f = f$ and 
   let us show that $f$ is indeed an $\IBnu$-algebra morphism from $(A, \out^{\istar})$
   to $(B, \out^{\iistar})$.
   By functoriality,
   \[
     \out \comp (\IBnu{f}) = (\IB{\id}{\IBnu{f}}) \comp
     (\IB{f}{\id}) \comp \out
   \]
   implies
   \[
     \out^{\iistar} \comp (\IBnu{f}) = \left( (\IB{f}{\id}) \comp
       \out \right)^{\iistar} = (f \bullet \out)^{\iistar}.
   \]
   But, by definition of complete Elgot $\IB{}{}$-algebra morphisms,
   \[
     (f \bullet \out)^{\iistar}=f \comp \out^{\istar},
   \]
   so the functor $G$ is well-defined. 
   
   To finish the proof, we need to show that both $G F$ and $F G$
   are identities. Since
   both functors act as the identity on morphisms, we only need to verify this on objects. 
   On the one hand, 
   \[
     (G F)(A, \chi) = G(A, \chi\comp \ext, \argument^\istar) = (A, \out^{\istar}) = (A,\chi),
   \]
   for $\out^{\istar}$ is defined as $\chi \comp \coit(\out)$ and $\coit(\out)$
   is the identity. Similarly,
   \[
     (F G)(A, a, \argument^\istar) = F(A, \out^{\istar}) = (A, \out^{\istar} \comp
     \ext, \argument^\iistar) = (A, a, \argument^\istar),
   \]
   since, by functoriality applied to $\out \comp \coit(e) =
   (\IB{\id}{\coit(e)}) \comp e$,
   \[
     e^{\iistar} = \out^{\istar} \comp \coit(e) = e^{\istar},
   \]
   whatever $e:X\to\IB{A}{X}$ is, and moreover 
   \begin{flalign*}
     &&\out^{\istar} \comp \ext
     =&~ a \comp (\IB{\id}{\out^{\istar}}) \comp \out \comp \ext & \by {solution}\\
     &&=&~ a \comp (\IB{\id}{\out^{\istar}}) \comp (\IB{\id}{\etan}) &\by{definition of~$\ext$} \\
     &&=&~ a. & \by{$\out^{\istar}$ is an $\IBnu$-algebra}
   \end{flalign*} 
\qed
\subsection*{Proof of Proposition~\ref{prop:elgot_to_weak}}
In this proof all morphisms and compositions are in the Kleisli
category of the complete Elgot monad $\BBT$, and we denote identity
morphisms by their (co)domain. Note that the codiagonal law can,
equivalently, be written as
\[
((Y + \nabla) e)^\istar = e^{\istar\istar}
\] 
for any $e: X \to Y + X + X$, where $\nabla = [\inl,\inr]$ is the codiagonal
(hence the name of the law).

Let $g: X \to Y + X$ and $f: Y \to Z + Y$ and form the following morphism
\[
  w= (Y+X \xrightarrow{[\inl, g]} Y+X \xrightarrow{f+X} Z + Y + X \xrightarrow{Z+\inl+\inr} Z + (Y+X) + (Y+X)).
\]
Now observe that the left-hand morphism of~\eqref{eq:waterfall2} is $((Z+\nabla) w)^\istar \inr$. By the codiagonal law we have
\[
((Z+\nabla) w)^\istar \inr = w^{\istar\istar}\inr.
\]
So it remains to prove that $w^{\istar\istar}\inr = f^\istar g^\istar$. Clearly, we have
\begin{equation}\label{eq:altw}
  w = ((Z+\inl)f + (Y+X)) (Y + \inr) [\inl,g].
\end{equation}
Since $((Y + \inr) [\inl,g])\inr= (Y+\inr)g$ we obtain by functoriality that 
\begin{equation}\label{eq:gdag}
  ((Y + \inr) [\inl,g])^\istar \inr = g^\istar.
\end{equation}
Now we compute
\begin{align*}
  w^\istar & = (Z+\inl)f ((Y+\inr)[\inl, g])^\istar & \by{\eqref{eq:altw} and naturality} \\
  & = (Z+\inl)f [Y, ((Y+\inr)[\inl, g])^\istar] (Y+\inr) [\inl, g] & \by{fixpoint} \\
  & = (Z+\inl)f [Y, ((Y+\inr)[\inl, g])^\istar\inr] [\inl, g] \\
  & = (Z+\inl)f [Y, g^\istar] [\inl, g] & \by{\eqref{eq:gdag}} \\
  & = (Z+\inl)f [[Y, g^\istar]\inl, [Y, g^\istar]g] \\
  & = (Z+\inl)f [Y, g^\istar] & \by{fixpoint} 
\end{align*}
Now observe that $(Z+\inl)f[Y,g^\istar] \inl = (Z+\inl) f$ so that functoriality gives us
\begin{equation}\label{eq:fdag}
  ((Z+\inl)f[Y,g^\istar])^\istar \inl = f^\istar.
\end{equation}
Finally, we compute
\begin{align*}
  w^{\istar\istar}\inr & = ((Z+\inl) f [Y,g^\istar])^\istar\inr \\
  & = [Z, ((Z+\inl) f [Y,g^\istar])^\istar] (Z+\inl) f [Y,g^\istar] \inr & \by{fixpoint} \\
  & = [Z, ((Z+\inl) f [Y,g^\istar])^\istar\inl] f g^\istar \\
  & = [Z, f^\istar] f g^\istar & \by{\eqref{eq:fdag}}\\
  & = f^\istar g^\istar & \by{fixpoint}
\end{align*}
This completes the proof.%
\qed

\subsection*{Proof of Theorem~\ref{thm:emon}}
First note that~\eqref{eq:waterfall2} can, equivalently, be rewritten as
\begin{align}\label{eq:waterfall}
([(\eta\oplus\kinl)f,\ul{\inr\inr}]\klcomp [\kinl, g])^\istar\inr = f^\istar\klcomp g^\istar
\end{align}
where $g:X\kto Y + X$, $f:Y\kto Z+Y$.

(i)~Let us prove the first clause. By Corollary~\ref{cor:T-Sigma-bialg}, $\BC^{\BBT}$
is isomorphic to the category of all those $\IB{}{}$-algebras whose structure
factor through $T\nabla$, specifically, every $\BBT$-algebra $(A,a:TA\to A)$ 
gives rise to a $\IB{}{}$-algebra $(A,a\comp (T\nabla):T(A+A)\to A)$. In the case 
at hand, we equip every such $(A,a\comp (T\nabla):T(A+A)\to A)$ with an iteration $\argument^\iistar$ 
operator sending any $e:X\to T(A+X)$ to $e^\iistar=a\comp e^\istar:X\to A$.

Let us check the axioms of complete Elgot $\IB{}{}$-algebras. 
\begin{citemize}
  \item \emph{Solution.} This follows easily from the \textit{fixpoint}
  property of the dagger of the complete Elgot monad $\BBT$:
  \begin{align*}
    e^{\iistar} =&~ a  \comp e^{\istar} \\
    =&~ a  \comp [\eta, e^{\istar}]^{\klstar} \comp e \\
    =&~ a  \comp \mu \comp T[\eta, e^{\istar}] \comp e \\
    =&~ a  \comp (Ta)  \comp T[\eta, e^{\istar}] \comp e \\
    =&~ a  \comp T[\id, e^{\iistar}] \comp e \\
    =&~ a  \comp (T\nabla) \comp T(\id + e^{\istar}) \comp e\\
    =&~ a \comp (T\nabla) (\IB{\id}{e^{\istar}}) \comp e.
  \end{align*}
\item\emph{Functoriality} is a trivial application of \textit{uniformity}:
  \[
    f \comp h = (\IB{\id}{h}) \comp e = T(\id + h) \comp e
  \]
  implies
  \[
    f^{\iistar} \comp h = a \comp f^{\istar} \comp h = a \comp e^{\istar} = e^{\iistar}.
  \]

\item\emph{Compositionality.} Since $\IB{X}{Y} = T(X + Y)$, we have
  \[
    f \filledsquare g = [T(\id + \inl) \comp f, \eta \inr
    \inr]^{\klstar} \comp [\eta  \inl, g] : Y + X \longrightarrow T(A + (Y
    + X)).
  \]
Hence, by~\eqref{eq:waterfall},
\[
(f \filledsquare g)^\istar \comp\inr = f^\istar \klcomp g^\istar.
\]
Composing with $a: TA \to A$ we obtain
\begin{align*}
(f \filledsquare g)^\iistar\comp\inr = a\comp (f^\istar\klcomp g^\istar).
\end{align*} 
Let us further rewrite the right hand side.
\begin{flalign*}
&&a\comp (f^\istar\klcomp g^\istar) 
=&\; a\comp\mu\comp  T(f^\istar) g^\istar\\
&&=&\; a\comp (Ta)\comp  T(f^\istar) g^\istar\\
&&=&\; a\comp T(a\comp f^\istar) g^\istar\\
&&=&\; a\comp (T(a\comp f^\istar+\id)\comp g)^\istar & \by{naturality}\\
&&=&\; ((\IB{f^\iistar}{\id})\comp g)^\iistar\\
&&=&\; (f^{\iistar} \bullet g)^{\iistar}.
\end{flalign*}
\end{citemize}
Next, let us show that any $\BBT$-algebra morphism $h:A\to B$ from $(A,a)$ to $(B,b)$ 
gives rise to a complete Elgot $\IB{}{}$-algebra morphism between the corresponding
$\IB{}{}$-algebras, i.e.~for every $f:X\to\IB{A}{X}$ we have
\begin{displaymath}
h\comp f^\iistar = ((\IB{h}{\id})\comp f)^\iistar.
\end{displaymath}
 Indeed,
\begin{flalign*}
&&h\comp f^\iistar=&\; h\comp a\comp f^\istar\\  
&&=&\; b\comp (Th)\comp f^\istar\\  
&&=&\; b\comp (T(h+\id)\comp f)^\istar & \by{naturality}\\ 
&&=&\; ((\IB{h}{\id})\comp f)^\iistar. 
\end{flalign*} 
We have constructed a functor from $\BC^\BBT$ to $\cbElg(\BC)$. This functor is
full and faithful because the underlying functor from $\BC^\BBT$ into the category
of $\IB{}{}$-algebras is full and faithful by Corollary~\ref{cor:T-Sigma-bialg} and
any morphism of complete $\IB{}{}$-algebras is a morphism of $\IB{}{}$-algebras by Proposition~\ref{prop:hom}.

Finally, let us check that any complete Elgot $\IB{}{}$-algebra of the form 
$(A,a\comp (T\nabla),\argument^\iistar)$ satisfying $e^\iistar = a\comp (T\nabla)\comp (T\inl)\comp e^\istar$ 
for every $e:X\to T(A+X)$, is an image of a $\BBT$-algebra, specifically
of $(A,a)$. We only have to verify that $(A,a)$ is indeed a $\BBT$-algebra. This
is however straightforwards from the axioms of $\IB{}{}$-algebras and the definitions
$u_X^Y=\eta\inl$ and $m_X^Y=[\id,\eta\inr]^\klstar$.

(ii)~We now proceed with the second clause. To that end we have to verify the axioms 
of weak complete Elgot monads.

\begin{citemize}
 \item\emph{Fixpoint.} Given $f:X\to T(Y+X)$, 
\begin{flalign*}
&& f^\istar 
 =&\; (T(\eta+\id)\comp f)^\iistar\\
&& =&\; \mu\comp (T\nabla)\comp T(\id+f^\istar)\comp T(\eta+\id)\comp f&\by{solution}\\
&& =&\; [\eta,f^\istar]^\klstar\comp f.
\end{flalign*}
 \item\emph{Naturality.} Let $f:X\to T(Y+X)$ and let $g:Y\to TZ$. We consider two
special cases $g=\id: TY \to TY$ and $g=\eta\comp h$, where $h:Y \to Z$, i.e.~we will prove
\begin{align}
(\eta\comp h)^\klstar\comp f^\istar =&\, ([\eta\inl\comp h,\eta\inr]^\klstar f)^\istar\label{eq:nat_eta}\\
\id^\klstar\comp f^\istar =&\, ([T\inl,\eta\inr]^\klstar f)^\istar,\label{eq:nat_id}
\end{align}
which jointly imply the general case as follows:
\begin{flalign*}
&&g \klcomp f^\istar = &\; g^\klstar\comp f^\istar\\ 
&&= &\; \mu\comp (T g)\comp f^\istar&\\ 
&&= &\; \mu\comp (\eta\comp g)^\klstar\comp f^\istar\\ 
&&= &\; \mu\comp ([\eta\inl\comp g,\eta\inr]^\klstar f)^\istar &\by{\eqref{eq:nat_eta}}\\ 
&&= &\; \id^\klstar \comp ([\eta\inl\comp g,\eta\inr]^\klstar f)^\istar \\
&&= &\; ([T\inl,\eta\inr]^\klstar\comp [\eta\inl\comp g,\eta\inr]^\klstar f)^\istar &\by{\eqref{eq:nat_id}}\\ 
&&= &\; ([(T\inl)\comp g,\eta\inr]^\klstar\comp f)^\istar\\
&&= &\; ((g \oplus \eta) \klcomp f)^\istar.
\end{flalign*}
The proof of~\eqref{eq:nat_eta} is based on the fact that $Th: TY \to TZ$ is a morphism of $\BBT$-algebras from $(TY,\mu)$ to $(TZ,\mu)$ and hence,
by assumption, $h$ is a morphism of $\IB{}{}$-algebras from $J(TY,\mu)$ to $J(TZ,\mu)$:
\begin{flalign*}
&&(\eta h) \klcomp f^\istar = &\; (\eta\comp h)^\klstar\comp f^\istar \\
&& =&\; (Th)\comp (T(\eta+\id)\comp f)^\iistar\\
&& =&\; (T(Th+\id)\comp T(\eta+\id)\comp f)^\iistar&\\ 
&& =&\; (T(\eta\comp h+\id)\comp f)^\iistar\\
&& =&\; (T(\eta+\id)\comp [\eta\inl\comp h,\eta\inr]^\klstar f)^\iistar\\
&& =&\; ([\eta\inl\comp (\eta h),\eta\inr]^\klstar f)^\istar\\
&& =&\; ((\eta h) \oplus \eta) \klcomp f)^\istar.
\end{flalign*}
Now we prove~\eqref{eq:nat_id}. In this case we have
$f: X \to T(TY + X)$. We apply compotitionality to $f$ and
$T(\inl \eta): TY \to T(TY + TY)$ to obtain
\begin{equation}\label{eq:aux}
  (T(\inl\eta) \filledsquare f)^{\iistar} \comp \inr = ((T(\inl\eta))^{\iistar} \bullet f)^{\iistar}.
\end{equation}

First of all, note that
\begin{flalign*}
&& (T(\inl\eta))^{\iistar} 
=&\; [\id,\eta\inr]^\klstar\comp T(\id+(T(\inl\eta))^{\iistar}) \comp T(\inl\eta)&\by{fixpoint}\\
&&=&\; [\id, \eta\inr\comp (T(\inl\eta))^{\iistar}]^\klstar \comp T(\inl\eta)\\
&&=&\; \eta^\klstar\\
&&=&\; \id.
\end{flalign*}
Using the fact that $\id^\klstar=\mu$ is a morphism of complete Elgot
$\IB{}{}$-algebras, we obtain the left hand side of~\eqref{eq:nat_id}
from the left-hand side of~\eqref{eq:aux}
\begin{flalign*}
&&((T(\inl\eta))^{\iistar} \bullet f)^{\iistar} =&\; (T((T(\inl\eta))^{\iistar}+\id)\comp f)^\iistar&\\
&&=&\;f^\iistar\\
&&=&\;(T(\mu+\id)\comp T(\eta+\id))\comp f^\iistar\\
&&=&\;\mu\comp (T(\eta+\id)\comp f)^\iistar\\
&&=&\;\id^\klstar\comp f^\istar.
\end{flalign*}
In order to prove that we obtain the right-hand side of~\eqref{eq:nat_id} from the right-hand side of~\eqref{eq:aux}, 
let us denote $T(\inl\eta) \filledsquare f:TY+X\to T(TY+(TY+X))$ by $t$. Then, by definition,
\begin{align*}
t &= m^{TY+X}_{TY} \comp (\IB{((\IB{\id}{\inl}) \comp T(\inl\eta))}{\inr}) \comp [u_{TY}^X, f]\\
&= \mu \comp T[\id,\eta\inr] \comp T(\underbrace{T(\id + \inl) \comp T(\inl\eta)}_{T(\inl\eta)} + \inr)[\eta\inl, f] \\
&= \mu \comp T[T(\inl\eta), \eta \inr \inr]\comp [\eta\inl,f] \\
&= [T(\inl\eta),\eta\inr\inr]^\klstar \comp [\eta\inl, f]\\
&= [T(\inl\eta), [T(\inl\eta),\eta\inr\inr]^\klstar\comp f].
\end{align*}

Observe that we have 
\begin{align*}
t\comp\inr &= [T(\inl\eta), \eta\inr\inr]^\klstar \comp f\\
&= T(\id + \inr)\comp [T(\inl\eta), \eta\inr]^\klstar f \\
&= (\IB{\id}{\inr})\comp [T(\inl\eta),\eta\inr]^\klstar\comp f.
\end{align*}
Therefore, using uniformity we obtain
\begin{align*}
t^\iistar\inr &= ([T(\inl\eta),\eta\inr]^\klstar\comp f)^\iistar\\
&= (T(\eta+\id)\comp [T\inl,\eta\inr]^\klstar\comp f)^\iistar\\
&= ([T\inl,\eta\inr]^\klstar\comp f)^\istar,
\end{align*}
which is the right hand side of~\eqref{eq:nat_id}.

 \item\emph{Uniformity.} Suppose, $f:X\to T(Y+X)$,  $g: Z \to T(Y + Z)$, $h:Z\to X$ and
$f\comp h=T(\id+h)\comp g$. The latter implies $T(\eta+\id)\comp f\comp h=T(\eta+\id)\comp T(\id+h)\comp g=
T(\id+h)\comp T(\eta+\id)\comp g$ and hence we can apply functoriality of $\argument^\iistar$:
\begin{flalign*}
&&f^\istar\comp h
=&\; (T(\eta+\id)\comp f)^\iistar\comp h\\
&&=&\; (T(\eta+\id)\comp g)^\iistar &\by{functoriality}\\
&&=&\;g^\istar
\end{flalign*}
as required.
\end{citemize}
Finally, let us check~\eqref{eq:waterfall}. We start with the following instance
of compositionality:
\begin{align*}
([T(\eta + \inl) \comp f, \eta \inr
    \inr]^{\klstar} \comp [\eta  \inl, g])^\iistar\inr = ((T(\eta + \id)f \filledsquare g)^\iistar = (T((T(\eta+\id) f)^\iistar+\id))\comp g)^\iistar
\end{align*}
where we used the assumption that $J(TZ,\mu)$ is a complete Elgot $\IB{}{}$-algebra. Now,
\begin{flalign*}
&&([T(\eta + \inl)& \comp f, \eta \inr
    \inr]^{\klstar} \comp [\eta  \inl, g])^\iistar\inr&\\
&&=&\;(T(\eta + \id)\comp [T(\id + \inl) \comp f, \eta \inr
    \inr]^{\klstar} \comp [\eta  \inl, g])^\iistar\inr\\
&&=&\;([T(\id + \inl) \comp f, \eta \inr
    \inr]^{\klstar} \comp [\eta  \inl, g])^\istar\inr\\
&&=&\; ([(\eta\oplus\kinl)f,\ul{\inr\inr}]\klcomp [\kinl, g])^\istar\inr \\[2ex]
&&(T((T(\eta+\id) \;&f)^\iistar+\id))\comp g)^\iistar\\
&&=&\;(T(f^\istar+\id)\comp g)^\iistar\\
&&=&\;(T(\mu+\id)\comp T((\eta T)+\id)\comp T(f^\istar+\id)\comp g)^\iistar\\
&&=&\;\mu\comp (T((\eta T)+\id)\comp T(f^\istar+\id)\comp g)^\iistar&\by{$\mu$ is a morphism in $\cbElg(\BC)$}\\
&&=&\;\mu\comp (\comp T(f^\istar+\id)\comp g)^\istar\\
&&=&\;\mu\comp (Tf^\istar)\comp g^\istar&\by{naturality of~$\argument^\istar$}\\
&&=&\;f^\istar\klcomp g^\istar
\end{flalign*}
which in summary yields~\eqref{eq:waterfall}.
\qed

\subsection*{Proof of Theorem~\ref{thm:elgot_from_alg}}
Suppose that $\BBT$ is a complete Elgot monad and let us show~\eqref{eq:codiag_alg}.
Note the identity
\begin{align}\label{eq:iistar_from_istar}
f^\iistar = a\comp f^\istar,
\end{align}
which holds for every $f:X\to\IB{A}{X}$ and every $\BBT$-algebra $(A,a)$ because $a:TA\to A$ 
is a morphism of $\BBT$-algebras from $(TA,\mu)$ to $(A,a)$, hence a morphism of complete Elgot $\IB{}{}$-algebras 
from $J(TA,\mu)$ to $J(A,a)$ and therefore we have
\begin{align*}
f^\iistar = ((\IB{a}{\id})\comp (\IB{\eta}{\id})\comp f)^\iistar = a\comp ((\IB{\eta}{\id})\comp f)^\iistar  = a\comp f^\istar.
\end{align*}
Now,~\eqref{eq:codiag_alg} is obtained as follows: for any $e: X \to \IB{(\IB{A}{X})}{X} = T(T(A + X) + X$ we have
\begin{flalign*}
&&(m^X_{\IB{A}{X}}\comp e)^\iistar =\;& ([\id,\eta\inr]^\klstar\comp e)^\iistar&\\
&&=\;& (T(a+\id)\comp [T(\eta+\id),\eta\inr]^\klstar\comp e)^\iistar\\
&&=\;&a\comp ([T(\eta+\id),\eta\inr]^\klstar\comp e)^\iistar & \by{$a$ preserves $\argument^\iistar$}\\
&&=\;&a\comp (T(\eta+\id)\comp [\id,\eta\inr]^\klstar\comp e)^\iistar\\
&&=\;&a\comp ([\id,\eta\inr]^\klstar\comp e)^\istar & \by{\eqref{eq:iistar_def}}\\
&&=\;&a\comp (T[\id,\inr]\comp [T\inl,\eta\inr]^\klstar\comp e)^\istar\\
&&=\;&a\comp (([T\inl,\eta\inr]^\klstar\comp e)^\istar)^\istar&\by{codiagonal}\\
&&=\;&a\comp (((\id \oplus \eta) \klcomp e)^\istar)^\istar \\
&&=\;&a\comp (\id^\klstar\comp e^\istar)^\istar&\by{naturality}\\
&&=\;&a\comp (\mu\comp e^\istar)^\istar\\
&&=\;&(e^\iistar)^\iistar.&\by{\eqref{eq:iistar_from_istar}}
\end{flalign*}
Conversely, we assume~\eqref{eq:codiag_alg} and prove that $\BBT$ is a complete Elgot monad.
By Theorem~\ref{thm:emon}, we only need to verify the codiagonal identity. Let 
$f:X\to T((Y+X)+X)$ and let us take 
\begin{align*}
e=T(\eta\comp (\eta+\id)+\id)\comp f:X\to T(T(TY + X) + X) = \IB{(\IB{TY}{X})}{X}
\end{align*}
in~\eqref{eq:codiag_alg}. Then we obtain the codiagonal identity for $f$ as follows:
\begin{flalign*}
&&(T[\id,\inr]\comp f)^\istar 
=\;& (T(\eta+\id)\comp T[\id,\inr]\comp f)^\iistar&\\
&&=\;& (T[\eta+\id,\inr]\comp f)^\iistar\\
&&=\;& ([\eta\comp(\eta+\id),\eta\inr]^\klstar\comp f)^\iistar\\
&&=\;& ([\id,\eta\inr]^\klstar\comp T(\eta\comp (\eta+\id)+\id)\comp f)^\iistar\\
&&=\;& ((T(\eta\comp (\eta+\id)+\id)\comp f)^\iistar)^\iistar&\by{\eqref{eq:codiag_alg}}\\
&&=\;& ((T(\eta+\id)\comp T((\eta+\id)+\id)\comp f)^\iistar)^\iistar\\
&&=\;& ((T((\eta+\id)+\id)\comp f)^\istar)^\iistar\\
&&=\;& (T(\eta+\id)\comp f^\istar)^\iistar&\by{naturality}\\
&&=\;& (f^\istar)^\istar.
\end{flalign*}
This completes the proof.
\qed

\fi 
\end{document}

%% file: catprog.tex
\providecommand{\catname}{\mathbf} 
\providecommand{\clsname}{\mathcal}

\def\defcatname#1{\expandafter\def\csname B#1\endcsname{\catname{#1}}}
\def\defcatnames#1{\ifx#1\defcatnames\else\defcatname#1\expandafter\defcatnames\fi}
\defcatnames ABCDEFGHIJKLMNOPQRSTUVWXYZ\defcatnames

\def\defclsname#1{\expandafter\def\csname C#1\endcsname{\clsname{#1}}}
\def\defclsnames#1{\ifx#1\defclsnames\else\defclsname#1\expandafter\defclsnames\fi}
\defclsnames ABCDEFGHIJKLMNOPQRSTUVWXYZ\defclsnames

\def\defbbname#1{\expandafter\def\csname BB#1\endcsname{\mathbb{#1}}}
\def\defbbnames#1{\ifx#1\defbbnames\else\defbbname#1\expandafter\defbbnames\fi}
\defbbnames ABCDEFGHIJKLMNOPQRSTUVWXYZ\defbbnames

\def\Set{\catname{Set}}

\def\Cppo{\catname{Cppo}}

\providecommand{\argument}{\operatorname{-\!-}}

\providecommand{\ul}{\underline}			

\providecommand{\PSet}{{\mathcal P}}			
\providecommand{\Id}{\operatorname{Id}}

\providecommand{\Hom}{\mathsf{Hom}}
\providecommand{\id}{\mathsf{id}}
\providecommand{\comp}{\mathbin{\circ}}

\providecommand{\inl}{\operatorname{\sf inl}}
\providecommand{\inr}{\operatorname{\sf inr}}

\providecommand{\dash}{\nobreakdash-\hspace{0pt}}	

\providecommand{\by}[1]{\text{/\!/~#1}}			
\providecommand{\pacman}[1]{}				

%% file: defslncs.tex
\spnewtheorem{thm}[theorem]{Theorem}{\bfseries}{\itshape}
\spnewtheorem{cor}[theorem]{Corollary}{\bfseries}{\itshape}
\spnewtheorem{cnj}[theorem]{Conjecture}{\bfseries}{\itshape}
\spnewtheorem{lem}[theorem]{Lemma}{\bfseries}{\itshape}
\spnewtheorem{lemdefn}[theorem]{Lemma and Definition}{\bfseries}{\itshape}
\spnewtheorem{prop}[theorem]{Proposition}{\bfseries}{\itshape}
\spnewtheorem{defn}[theorem]{Definition}{\bfseries}{\upshape}
\spnewtheorem{rem}[theorem]{Remark}{\bfseries}{\upshape}
\spnewtheorem{notation}[theorem]{Notation}{\bfseries}{\upshape}
\spnewtheorem{expl}[theorem]{Example}{\bfseries}{\upshape}
\spnewtheorem{thmdefn}[theorem]{Theorem and Definition}{\bfseries}{\itshape}
\spnewtheorem{propdefn}[theorem]{Proposition and Definition}{\bfseries}{\itshape}
\spnewtheorem{assumption}[theorem]{Assumption}{\bfseries}{\upshape}
\spnewtheorem{algorithm}[theorem]{Algorithm}{\bfseries}{\upshape}

%% file: diagram.tex
\begin{figure}[t!]
  \tikzset{
    font=\tiny,
    nonterminal/.style={
      rectangle,
      minimum size=6mm,
      very thick,
      draw=orange!50!black!50,         
      fill=orange!50!white,
      font=\itshape
    },
    terminal/.style={
      scale=.5,
      circle,
      inner sep=0pt,
      thin,draw=black!50,
      top color=white,bottom color=black!20,
      font=\ttfamily
    },
    iterated/.style={
      fill=green!20,
      thick,
      draw=green!50
    },
    natural/.style={
      circle,
      minimum size=4mm,
      inner sep=2pt,
      thin,draw=black!50,
      top color=white,bottom color=black!20,
      font=\ttfamily},
    skip loop/.style={to path={-- ++(0,#1) -| (\tikztotarget)}}
  }

  {
    \tikzset{nonterminal/.append style={text height=1.5ex,text depth=.25ex}}
    \tikzset{natural/.append style={text height=1.5ex,text depth=.25ex}}
  }
  \captionsetup[subfigure]{labelformat=empty,justification=justified,singlelinecheck=false}
  \pgfdeclarelayer{background}
  \pgfdeclarelayer{foreground}
  \pgfsetlayers{background,main,foreground}
    \begin{subfigure}{\textwidth}
    \centering
    \caption{Fixpoint:}
    \vspace{-2ex}
    \raisebox{-.5\height}{
    \begin{tikzpicture}[
      point/.style={coordinate},>=stealth',thick,draw=black!50,
      tip/.style={->,shorten >=0.007pt},every join/.style={rounded corners},
      hv path/.style={to path={-| (\tikztotarget)}},
      vh path/.style={to path={|- (\tikztotarget)}},
      text height=1.5ex,text depth=.25ex 
      ]
      \node [nonterminal] (f) {$f$};
      \draw [<-] (f.west) -- +(-1,0) node [midway,above] {$X$};
      \path [<-,draw] ($(f.west)+(-0.5,0)$) -- +(0,-0.8) -| ($(f.east)+(0.5,-0.15)$) node [pos=0.8,right] {$X$} -- +(-0.5,0);
      \draw [->] (f.east)++(0,0.15) -- +(1,0) node [midway,above] {$Y$};
      
      \begin{pgfonlayer}{background}
        \draw [iterated] ($(f.north west)+(-0.25,0.25)$) rectangle ($(f.south east)+(0.25,-0.25)$);
      \end{pgfonlayer}
    \end{tikzpicture}
    }
    ~~=~~
    \raisebox{-.5\height}{
    \begin{tikzpicture}[
      point/.style={coordinate},>=stealth',thick,draw=black!50,
      tip/.style={->,shorten >=0.007pt},every join/.style={rounded corners},
      hv path/.style={to path={-| (\tikztotarget)}},
      vh path/.style={to path={|- (\tikztotarget)}},
      text height=1.5ex,text depth=.25ex 
      ]
      \node [nonterminal] (f) {$f$};
      \node [nonterminal] (f2) at ($(f.east)+(1.5,-0.15)$) {$f$};
      \draw [<-] (f.west) -- +(-1,0) node [midway,above] {$X$};
      \draw [->] (f.east)++(0,-0.15) -- (f2.west) node [pos=0.25,below] {$X$};
      \path [<-,draw] (f2.west)++(-0.5,0) -- ++(0,-0.8) -| ($(f2.east)+(0.5,-0.15)$) node [near end,right] {$X$} -- +(-0.5,0);
      \draw [->] (f.east)++(0,0.15) -- node [midway,above] {$Y$} ++(0.5,0) -- ++(0,0.5) -- ++(3,0);
      \draw [->] (f2.east)++(0,0.15) -- ++(0.5,0) -- node [midway,right] {$Y$} ++(0,0.65);

      \begin{pgfonlayer}{background}
        \draw [iterated] ($(f2.north west)+(-0.25,0.25)$) rectangle ($(f2.south east)+(0.25,-0.25)$);
      \end{pgfonlayer}
    \end{tikzpicture}
    }
  \end{subfigure}
  \par
  \begin{subfigure}{\textwidth}
    \centering
    \caption{Naturality:}
    \raisebox{-.5\height}{
    \begin{tikzpicture}[
      point/.style={coordinate},>=stealth',thick,draw=black!50,
      tip/.style={->,shorten >=0.007pt},every join/.style={rounded corners},
      hv path/.style={to path={-| (\tikztotarget)}},
      vh path/.style={to path={|- (\tikztotarget)}},
      text height=1.5ex,text depth=.25ex 
      ]
      \node [nonterminal] (f) {$f$};
      \node [nonterminal] (g) at ($(f.east)+(1.5,0.15)$) {$g$};
      \draw [<-] (f.west) -- +(-1,0) node [midway,above] {$X$};
      \path [<-,draw] (f.west)++(-0.5,0) -- ++(0,-0.8) -| ($(f.east)+(0.5,-0.15)$) node [near end,right] {$X$}  -- ++(-0.5,0);
      \draw [->] ($(f.east)+(0,0.15)$) -- (g) node [midway,above] {$Y$};
      \draw [->] (g.east) -- +(1,0) node [midway,above] {$Z$};
      \begin{pgfonlayer}{background}
        \draw [iterated] ($(f.north west)+(-0.25,0.25)$) rectangle ($(f.south east)+(0.25,-0.25)$);
      \end{pgfonlayer}
    \end{tikzpicture}
    }
    ~~=~~
    \raisebox{-.5\height}{
    \begin{tikzpicture}[
      point/.style={coordinate},>=stealth',thick,draw=black!50,
      tip/.style={->,shorten >=0.007pt},every join/.style={rounded corners},
      hv path/.style={to path={-| (\tikztotarget)}},
      vh path/.style={to path={|- (\tikztotarget)}},
      text height=1.5ex,text depth=.25ex 
      ]
      \node [nonterminal] (f) {$f$};
      \node [nonterminal] (g) at ($(f.east)+(1.5,0.15)$) {$g$};
      \draw [<-] (f.west) -- +(-1,0) node [midway,above] {$X$};
      \path [<-,draw] (f.west)++(-0.5,0) 
        -- ++(0,-0.75) 
        -| ($(g.east)+(0.5,-0.5)$) node [near end,right] {$X$} 
        -| ($(f.east)+(0.6,-0.15)$) 
        -- ++(-0.6,0);
      \draw [->] ($(f.east)+(0,0.15)$) -- (g) node [midway,above] {$Y$};
      \draw [->] (g.east) -- +(1,0) node [midway,above] {$Z$};
      \begin{pgfonlayer}{background}
        \draw [iterated] ($(f.north west)+(-0.25,0.25)$) rectangle ($(g.south east)+(0.25,-0.35)$);
      \end{pgfonlayer}
    \end{tikzpicture}
    }
  \end{subfigure}
  \par\medskip
  \begin{subfigure}{\textwidth}
    \centering
    \vspace{1ex}
    \caption{Codiagonal:}
    \vspace{-1ex}
    \raisebox{-.5\height}{
    \begin{tikzpicture}[
      point/.style={coordinate},>=stealth',thick,draw=black!50,
      tip/.style={->,shorten >=0.007pt},every join/.style={rounded corners},
      hv path/.style={to path={-| (\tikztotarget)}},
      vh path/.style={to path={|- (\tikztotarget)}},
      text height=1.5ex,text depth=.25ex 
      ]
      \node [nonterminal,minimum height=1.2cm] (f) {$g$};
      \draw [<-] (f.west) -- +(-1,0) node [midway,above] {$X$};
      \draw [->] (f.east)++(0,0.4) -- ++(1.5,0) node [pos=0.3,above] {$Y$};
      \draw [->] (f.east) -- ++(1,0) --node [midway,right] {$X$}  ++(0,-1.1) -| ($(f.west)+(-0.5,0)$);
      \draw [->] (f.east)++(0,-0.4) -- node [midway,below] {$X$} ++(0.5,0) -- ++(0,0.4);
      \begin{pgfonlayer}{background}
        \draw [iterated] ($(f.north west)+(-0.25,0.25)$) rectangle ($(f.south east)+(0.75,-0.25)$);
      \end{pgfonlayer}
    \end{tikzpicture}
    }
    ~~=~~
    \raisebox{-.5\height}{
    \begin{tikzpicture}[
      point/.style={coordinate},>=stealth',thick,draw=black!50,
      tip/.style={->,shorten >=0.007pt},every join/.style={rounded corners},
      hv path/.style={to path={-| (\tikztotarget)}},
      vh path/.style={to path={|- (\tikztotarget)}},
      text height=1.5ex,text depth=.25ex 
      ]
      \node [nonterminal,minimum height=1cm] (f) {$g$};
      \draw [<-] (f.west) -- +(-1.6,0) node [pos=0.7,above] {$X$};
      \draw [->] (f.east)++(0,0.3) -- ++(1.5,0) node [pos=0.3,above] {$Y$};
      \draw [->] (f.east) -- ++(1.1,0) -- node [midway,right] {$X$} ++(0,-1.35) -| ($(f.west)+(-0.95,0)$);
      \draw [->] (f.east)++(0,-0.3) -- ++(0.5,0) -- node [midway,right] {$X$} ++(0,-0.65) -| ($(f.west)+(-0.5,0)$);
      \begin{pgfonlayer}{background}
        \draw [iterated] ($(f.north west)+(-0.75,0.35)$) rectangle ($(f.south east)+(0.85,-0.6)$);
        \draw [iterated,fill=green!35] ($(f.north west)+(-0.25,0.2)$) rectangle ($(f.south east)+(0.25,-0.25)$);
      \end{pgfonlayer}
    \end{tikzpicture}
    }
  \end{subfigure}
  \par
  \begin{subfigure}{\textwidth}
    \caption{Uniformity:}
    \centering
\begin{tabular}{rcl}
   \raisebox{-.5\height}{
    \begin{tikzpicture}[
      point/.style={coordinate},>=stealth',thick,draw=black!50,
      tip/.style={->,shorten >=0.007pt},every join/.style={rounded corners},
      hv path/.style={to path={-| (\tikztotarget)}},
      vh path/.style={to path={|- (\tikztotarget)}},
      text height=1.5ex,text depth=.25ex 
      ]
      \node [nonterminal,fill=blue!20,draw=blue!50] (h) {$h$};
      \node [nonterminal] (f) at ($(h.east)+(1.5,0)$) {$f$};
      \draw [<-] (h.west) -- +(-1,0) node [midway,above] {$Z$};
      \draw [->] (h.east) -- (f.west) node [midway,above] {$X$};
      \draw [->] (f.east)++(0,0.15) -- +(1,0) node [midway,above] {$Y$};
      \draw [->] (f.east)++(0,-0.15) -- +(1,0) node [midway,below] {$X$};
    \end{tikzpicture}
    }
    &$~~=~~$&
    \raisebox{-.5\height}{
    \begin{tikzpicture}[
      point/.style={coordinate},>=stealth',thick,draw=black!50,
      tip/.style={->,shorten >=0.007pt},every join/.style={rounded corners},
      hv path/.style={to path={-| (\tikztotarget)}},
      vh path/.style={to path={|- (\tikztotarget)}},
      text height=1.5ex,text depth=.25ex 
      ]
      \node [nonterminal] (f) {$g$};
      \node [nonterminal,fill=blue!20,draw=blue!50] (h) at ($(f.east)+(1.5,-0.15)$) {$h$};
      \draw [<-] (f.west) -- +(-1,0) node [midway,above] {$Z$};
      \draw [->] (f.east)++(0,-0.15) -- (h.west) node [midway,below] {$Z$};
      \draw [->] (f.east)++(0,0.15) -- ++(0.65,0) -- ++(0,0.4) node [midway,left] {$Y$} -- ++(2.35,0);
      \draw [->] (h.east) -- +(1,0) node [midway,above] {$X$};
    \end{tikzpicture}
    }\\
    &$\Downarrow$&\\
    \raisebox{-.5\height}{
    \begin{tikzpicture}[
      point/.style={coordinate},>=stealth',thick,draw=black!50,
      tip/.style={->,shorten >=0.007pt},every join/.style={rounded corners},
      hv path/.style={to path={-| (\tikztotarget)}},
      vh path/.style={to path={|- (\tikztotarget)}},
      text height=1.5ex,text depth=.25ex 
      ]
      \node [nonterminal,fill=blue!20,draw=blue!50] (h) {$h$};
      \node [nonterminal] (f) at ($(h.east)+(1.5,0)$) {$f$};
      \draw [<-] (h.west) -- +(-1,0) node [midway,above] {$Z$};
      \draw [->] (h.east) -- (f.west) node [midway,above] {$X$};
      \draw [->] (f.east)++(0,0.15) -- +(1,0) node [midway,above] {$Y$};
      \path [<-,draw] (f.west)++(-0.5,0) -- ++(0,-0.8) -|
      ($(f.east)+(0.5,-0.15)$) node [pos=0.8,right] {$X$} -- ++(-0.5,0);
      \begin{pgfonlayer}{background}
        \draw [iterated] ($(f.north west)+(-0.25,0.25)$) rectangle ($(f.south east)+(0.25,-0.25)$);
      \end{pgfonlayer}
    \end{tikzpicture}
    }
    &$~~=~~$&
    \raisebox{-.5\height}{
    \begin{tikzpicture}[
      point/.style={coordinate},>=stealth',thick,draw=black!50,
      tip/.style={->,shorten >=0.007pt},every join/.style={rounded corners},
      hv path/.style={to path={-| (\tikztotarget)}},
      vh path/.style={to path={|- (\tikztotarget)}},
      text height=1.5ex,text depth=.25ex 
      ]
      \node [nonterminal] (f) {$g$};
      \draw [<-] (f.west) -- +(-1,0) node [midway,above] {$Z$};
      \path [<-,draw] ($(f.west)+(-0.5,0)$) -- +(0,-0.8) -| ($(f.east)+(0.5,-0.15)$) node [pos=0.8,right] {$Z$} -- +(-0.5,0);
      \draw [->] (f.east)++(0,0.15) -- +(1,0) node [midway,above] {$Y$};
      
      \begin{pgfonlayer}{background}
        \draw [iterated] ($(f.north west)+(-0.25,0.25)$) rectangle ($(f.south east)+(0.25,-0.25)$);
      \end{pgfonlayer}
    \end{tikzpicture}
    }
  \end{tabular}
  \end{subfigure}
\vspace{2ex}
\caption{Axioms of complete Elgot monads.}
\label{fig:ax}
\end{figure}


%% file: dinaturality.tex
\begin{figure}[t!]
  \tikzset{
    font=\tiny,
    nonterminal/.style={
      rectangle,
      minimum size=6mm,
      very thick,
      draw=orange!50!black!50,         
      fill=orange!50!white,
      font=\itshape
    },
    terminal/.style={
      scale=.5,
      circle,
      inner sep=0pt,
      thin,draw=black!50,
      top color=white,bottom color=black!20,
      font=\ttfamily
    },
    iterated/.style={
      fill=green!20,
      thick,
      draw=green!50
    },
    natural/.style={
      circle,
      minimum size=4mm,
      inner sep=2pt,
      thin,draw=black!50,
      top color=white,bottom color=black!20,
      font=\ttfamily},
    skip loop/.style={to path={-- ++(0,#1) -| (\tikztotarget)}}
  }

  {
    \tikzset{nonterminal/.append style={text height=1.5ex,text depth=.25ex}}
    \tikzset{natural/.append style={text height=1.5ex,text depth=.25ex}}
  }
  \captionsetup[subfigure]{labelformat=empty,justification=justified,singlelinecheck=false}
  \pgfdeclarelayer{background}
  \pgfdeclarelayer{foreground}
  \pgfsetlayers{background,main,foreground}
  \begin{subfigure}{\textwidth}
    \centering
    \vspace{2.5ex}
    \raisebox{-.5\height}{
    \begin{tikzpicture}[
      point/.style={coordinate},>=stealth',thick,draw=black!50,
      tip/.style={->,shorten >=0.007pt},every join/.style={rounded corners},
      hv path/.style={to path={-| (\tikztotarget)}},
      vh path/.style={to path={|- (\tikztotarget)}},
      text height=1.5ex,text depth=.25ex 
      ]
      \node [nonterminal] (f) {$g$};
      \node [nonterminal] (g) at ($(f.east)+(1.5,-0.15)$) {$h$};
      \draw [<-] (f.west) -- +(-1,0) node [midway,above] {$X$};
      \path [<-,draw] (f.west)++(-0.5,0) -- ++(0,-0.95) -| ($(g.east)+(0.85,-0.15)$) node [near end,right] {$X$} -- ++(-0.85,0);
      \draw [->] (f.east)++(0,0.15) -- ++(0.65,0) node [midway,above] {$Y$} -- ++(0,0.4) -- ++(2.5,0);
      \draw [->] (f.east)++(0,-0.15) -- (g.west) node [midway,below] {$Z$};
      \draw [->] (g.east)++(0,0.15) -- ++(0.25,0) -- ++(0,0.55) node [midway,right] {$Y$};

      \begin{pgfonlayer}{background}
        \draw [iterated] ($(f.north west)+(-0.25,0.5)$) rectangle ($(g.south east)+(0.6,-0.25)$);
      \end{pgfonlayer}
    \end{tikzpicture}
    }
    ~~=~~
    \raisebox{-.5\height}{
    \begin{tikzpicture}[
      point/.style={coordinate},>=stealth',thick,draw=black!50,
      tip/.style={->,shorten >=0.007pt},every join/.style={rounded corners},
      hv path/.style={to path={-| (\tikztotarget)}},
      vh path/.style={to path={|- (\tikztotarget)}},
      text height=1.5ex,text depth=.25ex 
      ]
      \node [nonterminal] (h) at ($(f.west)+(-1.2,0.15)$) {$g$};
      \node [nonterminal] (f) {$h$};
      \node [nonterminal] (g) at ($(f.east)+(1.5,-0.15)$) {$g$};
      \draw [<-] (h.west) -- ++(-0.5,0) node [midway,above] {$X$};
      \draw [->] (h.east)++(0,0.15) -- ++(0.4,0) -- ++(0,0.75) node [midway,left] {$Y$} -- ++(3.75,0) -- ++(0,-0.5);
      \draw [<-] (f.west) -- ($(h.east)+(0,-0.15)$) node [pos=0.8,below] {$Z$};
      \path [<-,draw] (f.west)++(-0.5,0) -- ++(0,-0.95) -| ($(g.east)+(0.85,-0.15)$) node [near end,right] {$Z$} -- ++(-0.85,0);
      \draw [->] (f.east)++(0,0.15) -- ++(0.65,0) node [midway,above] {$Y$} -- ++(0,0.4) -- ++(2.5,0);
      \draw [->] (f.east)++(0,-0.15) -- (g.west) node [midway,below] {$X$};
      \draw [->] (g.east)++(0,0.15) -- ++(0.25,0) -- ++(0,0.55) node [midway,right] {$Y$};

      \begin{pgfonlayer}{background}
        \draw [iterated] ($(f.north west)+(-0.25,0.5)$) rectangle ($(g.south east)+(0.6,-0.25)$);
      \end{pgfonlayer}
    \end{tikzpicture}
    }
  \end{subfigure}
  \par\medskip
\caption{Dinaturality axiom.}
\label{fig:dina}
\end{figure}
